\documentclass[a4paper, 12pt]{amsart}
\usepackage{amsmath}
\usepackage{graphicx}
\usepackage{amsfonts}
\usepackage{amssymb}

\newcommand{\ftp}{\mathbin{\raisebox{-1.4pt}{\includegraphics[height=1.75ex]{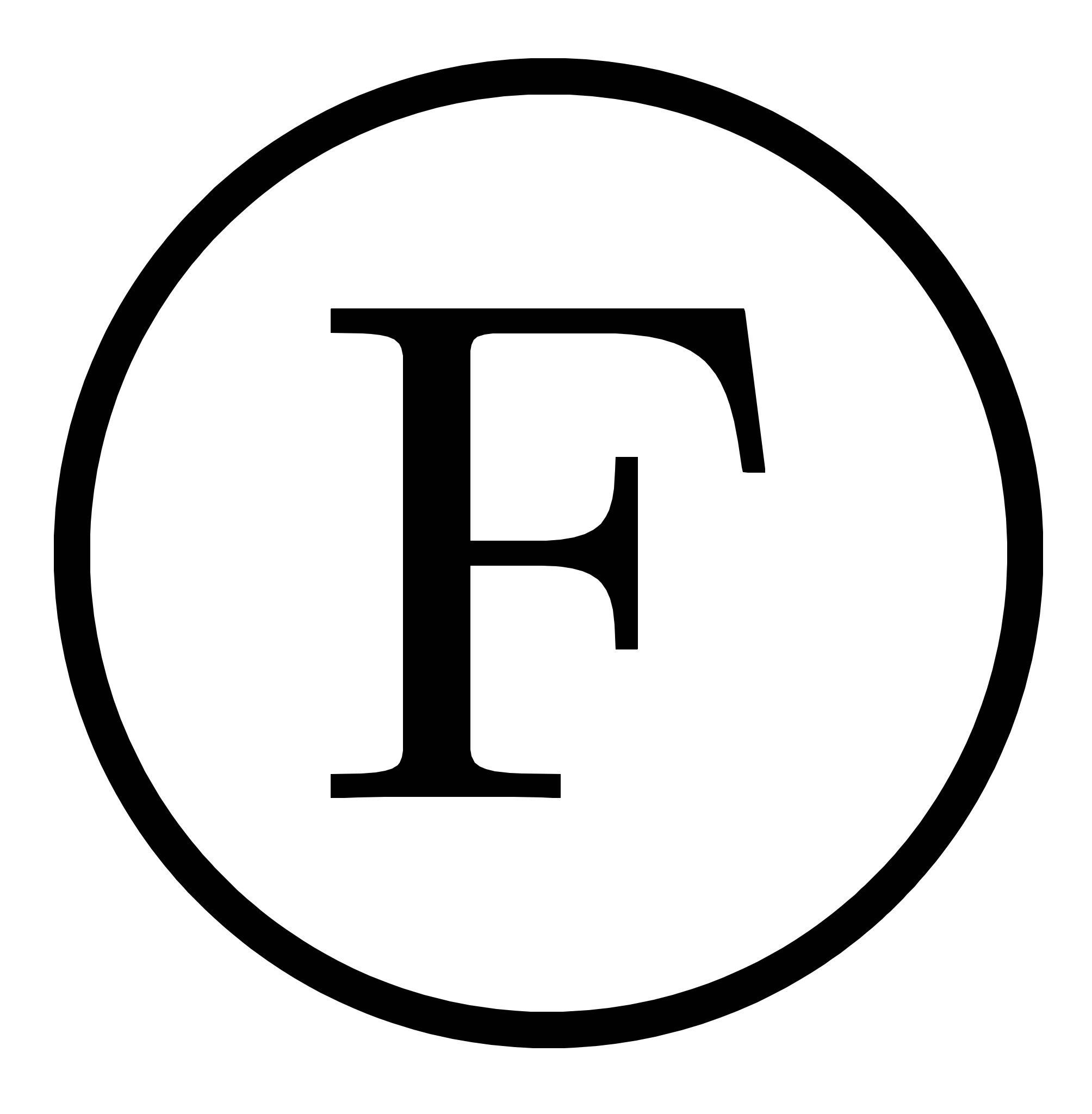}}}}

\newcommand{\ftps}{\mathbin{\raisebox{-1.2pt}{\includegraphics[height=1.25ex]{f1.png}}}}

\newcommand{\Aut}{\operatorname{Aut}}

\newcommand{\id}{\operatorname{id}}

\newcommand{\F}{\operatorname{F}}
\newcommand{\Us}{\operatorname{U}} 

\newtheorem{theorem}{Theorem}[section]
\newtheorem*{theorem*}{Theorem}
\newtheorem{cor}[theorem]{Corollary}
\newtheorem{lem}[theorem]{Lemma}
\newtheorem{prop}[theorem]{Proposition}

\newtheorem{proposition}[theorem]{Proposition}

\newtheorem{corollary}[theorem]{Corollary}

\theoremstyle{remark}
\newtheorem{remark}[theorem]{Remark}
\newtheorem{rem}[theorem]{Remark}

\theoremstyle{definition}
\newtheorem{definition}[theorem]{Definition}
\newtheorem*{definition*}{Definition}

\newtheorem{exa}[theorem]{Example}
\newtheorem{notat}[theorem]{Notation}
\newtheorem{conv}[theorem]{Convention}

\newtheorem*{Rev&OT}{Reversibility and optimal transport}

\newtheorem*{KMS}{KMS duals}

\newtheorem*{tra}{A lattice example}

\newtheorem*{RD}{Reversing dynamics}

\newtheorem*{parity}{Parity}

\newtheorem*{sf*}{Standard form}

\theoremstyle{definition}

\def\lp{\leftarrow}

\def\va{\mathbf{A}}
\def\vb{\mathbf{B}}

\def\tr{\mathop{\rm Tr}}

\def\br{{\mathbb R}}

\def\bz{{\mathbb Z}}

\def\ca{{\mathcal A}}
\def\cb{{\mathcal B}}
\def\cc{{\mathcal C}}

\def\cf{{\mathcal F}}
\def\cg{{\mathcal G}}
\def\ch{{\mathcal H}}

\def\cp{{\mathcal P}}

\def\gt{{\mathfrak T}}

\def\a{\alpha}
\def\b{\beta}
\def\g{\gamma}  
\def\d{\delta}  \def\D{\Delta}

\def\io{\iota}
\def\la{\lambda} 
\def\L{\Lambda}

\def\vk{\varkappa}

\def\m{\mu}
\def\p{\pi}
\def\n{\nu}
\def\r{\rho}
\def\s{\sigma} 

\def\f{\varphi}  
\def\th{\theta} 
\def\om{\omega} \def\Om{\Omega}

\def\u{\upsilon} \def\U{\Upsilon}

\begin{document}
	\date{2025-10-03}
\title{Fermionic optimal transport}
\author{Rocco Duvenhage, Dylan van Zyl and Paola Zurlo}
\address{Department of Physics\\
	University of Pretoria\\
	Pretoria 0002\\
	South Africa}
	\email{rocco.duvenhage@up.ac.za}
\address{Department of Physics\\
	University of Pretoria\\
	Pretoria 0002\\
	South Africa}
	\email{dylanvzyl@gmail.com}
\address{Dipartimento di Matematica\\
	Universit\`a degli studi di Bari\\
	Via E. Orabona, 4\\
	70125 Bari\\
	Italy}
    \email{paola.zurlo@uniba.it}

\begin{abstract}
	Quadratic Wasserstein distances are obtained between dynamical systems (with states as special case), on $\bz_2$-graded von Neumann algebras.
	This is achieved through a systematic translation from non-graded to $\bz_2$-graded transport plans, on usual and fermionic (or $\bz_2$-graded) tensor products respectively.
	The metric properties of these fermionic Wasserstein distances are shown, and their symmetries relevant to deviation of a system from quantum detailed balance are investigated.
	The latter is done in conjunction with the development of a complete mathematical framework for detailed balance in systems involving indistinguishable fermions.
\end{abstract}

\maketitle
\tableofcontents

\newpage

\section{Introduction}

This paper develops a theory of quadratic Wasserstein distances between dynamical systems on $\bz_2$-graded von Neumann algebras, using \emph{fermionic} transport plans. The latter take the grading into account, being defined on a $\bz_2$-graded, or fermionic, tensor product of the one system's von Neumann algebra  and the twisted commutant of the other.
The twisted commutant is obtained from the commutant and the Klein isomorphism associated to the grading.
From a physical point of view, a transport plan is a state on a composite system, which in this $\bz_2$-graded setup can for example include indistinguishable fermions accounted for by the tensor product structure mentioned above.

This builds on recent work \cite{D22, D23, DSS} on Wasserstein distances between quantum dynamical systems on von Neumann algebras in the usual case, that is, without gradings. One of the motivations for that work was to develop a new framework for the study of steady state non-equilibrium systems in quantum statistical mechanics, in terms of distances of a system from systems satisfying detailed balance.


The basic strategy and core technical problem of the paper is to translate fermionic transport plans to usual transport plans on the corresponding usual tensor product and using the commutant.
This has to be done within appropriate cyclic representations associated with the transport plans.

Before outlining the paper, some current and historical context.
A host of different approaches to (and applications of) noncommutative or quantum optimal transport distances have appeared, with the term ``distance" being used in a very broad sense, for example often not required to satisfy the triangle inequality.
These include (roughly in chronological order)
\cite{Con89,
	ZS,
	Rief99,
	BV,
	CM1,
	NG,
	NGT,
	AF,
	CGT, 
	GMP,
	H1,
	W,
	H2,
	DMTL,
	dePT,
	JLS22,
	CEFZ,
	GR,
	BS,
	HCW,
	AJW}
and subsequent papers building on these. 
This list of papers is certainly not exhaustive, but hopefully fairly representative.
They typically adapt or are inspired by one of the classical (i.e., not quantum) approaches to optimal transport, but the oldest, by Connes, appeared in relation to developing a general theory of noncommutative geometry.

A very brief and incomplete history of the classical case is represented by the references (chronologically)
\cite{Mon,
	Kan42,
	Kan48,
	Vas,
	BB}.
The books \cite{V1, V2} provide an account of much of the classical theory.
Terminology varies in the literature, with other (and historically better motivated) names for Wasserstein distance being Kantorovich distance, Kantorovich-Rubinstein distance and Monge-Kantorovich distance.

Among the above, the approach of \cite{dePT} is in spirit closest to the approach taken here, through its use of quantum channels to represent transport plans. However, our setting of standard forms of von Neumann algebras, exploiting commutants, is quite different.


As for the specific goals and structure of the paper, in Section \ref{duality} we briefly review essential background, including gradings, Klein isomorphisms, twisted commutants and various duals of unital positive maps, also setting up some conventions and notation for the paper.
In addition, certain technical and conceptual points related to twisted and KMS duals are addressed.
Section \ref{SecTrPl&Bal} builds directly on Section \ref{duality} to introduce the definitions of fermionic transport plans and $\bz_2$-graded systems, serving as the most basic 
ingredients of the paper.


Our entire development of fermionic Wasserstein distances relies heavily on cyclic representations of tensor product algebras obtained from transport plans on these tensor products. Both the Fermi and usual tensor products are relevant in this development. 
Section \ref{SecCR} sets up these representations.
Section \ref{FTPvsTP} then proceeds in this setting, establishing a one-to-one correspondence (Theorem \ref{Thm1-to-1}) between fermionic and usual transport plans.
This section is the heart of the paper's approach.

The cost of a fermionic transport plan, as well as fermionic Wasserstein distances, are defined in Section \ref{SecFW2}.
The theory developed in Section \ref{FTPvsTP} allows one to derive metric properties of fermionic Wasserstein distances by invoking the corresponding results for the usual case, established in \cite{D22, D23, DSS}.

Lastly, in Section 7, the deviation of a system from detailed balance in terms of Wasserstein distance is investigated in a fermionic context.
This requires first developing the theory of fermionic detailed balance beyond the rudiments of the theory put in place by \cite{D18} (in an elementary context) and \cite{CDF} (in an operator algebra framework). It includes addressing foundational aspects of quantum detailed balance.
Subsequently, symmetries (or isometries) of fermionic Wasserstein distances relevant to fermionic detailed balance are derived. These are then used to obtain inequalities placing bounds on the deviation of a system from fermionic detailed balance, in terms of the system's Wasserstein distance from another that does satisfy fermionic detailed balance.

The paper includes a short appendix setting out some aspects of standard forms and Tomita-Taksesaki theory, 
as these are at the foundation of the entire development.

\section{Duality}

\label{duality}

In order to develop the theory of fermionic transport plans in the next section, we first recall
some theory regarding the duals of even (or grading-equivariant) unital positive maps.
In particular, the first part of this section introduces basic notation and conventions to be used throughout, and summarizes essential points about twisted commutants,
which are treated in \cite{DHR1, DHR2}, and also discussed for example in \cite{CDF}. 
Then we provide a treatment of the relationship between different duals, including the KMS dual.  In the process, we begin to explain why it is natural to couch deviation from detailed balance in an optimal transport framework.


Consider two $\s$-finite von Neumann algebras $\ca$
and $\cb$, both in standard form (see the Appendix, including for some notation used below) on the Hilbert spaces $\cg_\ca$ and $\cg_\cb$ respectively.
Denote the set of all faithful normal states on $\ca$ by
$$
\cf(\ca).
$$ 

For any $\m  \in\cf(\ca)$, there is a uniquely determined cyclic and separating vector $\L_\m$ in the natural positive cone $\cp_\ca \subset \cg_\ca$ of the standard form, such that 
$$
\mu(a)=\left\langle \L_\m , a\L_\m \right\rangle
$$
for all $a \in \ca$, which we remind the reader allows us to define a state 
$\m' \in \cf(\ca')$ on the
commutant $\ca'$  of $\ca$ by
$$
\m'(a') = \m \circ j_\ca (a') = \left\langle \L_\m, a' \L_\m \right\rangle
$$
for all $a' \in \ca'$,
where $j_\ca = J_\ca (\cdot)^* J_\ca$ on $B(\cg_\ca)$ in terms of the modular conjugation $J_\ca$ as in the Appendix.

Assume, in addition, that $\ca$ and $\cb$ have $\bz_2$-gradings 
$$
\g_\ca\text{ and } \g_\cb\, ,
$$
which are involutive $*$-automorphisms of $\ca$ and $\cb$ respectively. 
Denote the set of even faithful normal states on $\ca$ by
$$\cf_+ (\ca) \, ,$$
recalling that \emph{even} means
$\mu\circ\gamma_\ca=\mu$.
Note that 
$$\cf_+ (\ca) \neq \varnothing \, ,$$
since for any $\mu \in \cf (\ca)$, which is non-empty as $\ca$ is $\s$-finite, we have $\frac{1}{2} (\mu + \mu\circ\g_\ca) \in \cf_+ (\ca).$
The symbol $\nu\in\cf_+(\cb)$ will be used for states on $\cb$.
For $a\in\ca$, we'll also write
\begin{equation*}
	a_{\pm} :=  \frac{1}{2} \left( a \pm \g_\ca (a) \right) .
\end{equation*}
By the theory of standard forms there is a unique unitary representation 
$$g_\ca$$  
of $\g_\ca$ on $\cg_\ca$ 
satisfying 
$g_\ca a\L_\m = \g_\ca (a) \L_\m$
for all $\m \in \cf_+ (\ca)$,
in particular leaving each of the cyclic vectors $\L_\m$  invariant.
In addition,
\begin{equation}
	\label{Jg=gJ}
	J_\ca g_\ca = g_\ca J_\ca \, .
\end{equation}
See \cite[Corollary 2.5.32]{BR1} for a convenient presentation of all this.

This allows one to define the \emph{twisted commutant} of $\ca$ as
$$\ca^\wr= g_\ca^{1/2} \ca' g_\ca^{-1/2}$$ 
in terms of the unitary operator
$$g_\ca^{1/2}:= p_\ca^+ -ip_\ca^-$$
and its adjoint 
$$g_{\ca}^{-1/2} = p_\ca^+ + ip_\ca^-$$
obtained from
$$
p_\ca^{\pm} := \frac{1}{2}(1 \pm g_\ca) \, .
$$
It is easily checked 
that 
$$g_\ca^{1/2}g_\ca^{1/2} = g_\ca \, .$$
Note that $\L_\m$ is cyclic and separating for $\ca^\wr$.

Because of the notational complexity related to representations obtained from transport plans (see Section \ref{SecCR}), the Klein isomorphism associated to $\g_\ca$ 
(the idea behind it tracing back to \cite{JW, Kl38})
will be denoted by
$$
\g_{\ca}^{1/2} : \ca' \to \ca^\wr : 
a' \mapsto g_\ca^{1/2} a' {g_{\ca}^{-1/2}},
$$
its inverse being denoted by
$\g_{\ca}^{-1/2}$.
Using this $*$-isomorphism between $\ca'$ and $\ca^\wr$,  one can carry or copy $\m'$ to the latter by
\begin{equation} \label{verwToest}
	\mu^{\wr}(a^{\wr})
	=
	\m' \circ \g_{\ca}^{-1/2} (a^{\wr}) 
	= 
	\m \circ j_\ca \circ \g_{\ca}^{-1/2} (a^{\wr})
	=
	\left\langle \Lambda_{\mu}, a^{\wr}\Lambda_{\mu }\right\rangle
\end{equation}
for any  $\mu\in\cf_+ (\ca)$ and all $a^{\wr}\in \ca^{\wr}$ .

Of course, $\g_{\ca}$ can be extended to $B(\cg_{\ca})$ using its unitary representation $g_\ca$ above, but that would cause a lack of notational clarity later on.
Instead, we rather explicitly define the resulting $\bz_2$-gradings of $\ca'$ and $\ca^\wr$ separately by
$$
\g_{\ca'} (a') = g_\ca a' g_\ca 
\quad \text{and} \quad 
\g_{\ca^\wr} (a^\wr) = g_\ca a^\wr g_\ca
$$
for all $a' \in \ca'$ and $a^\wr \in \ca^\wr$.

	
Given this background, we now turn to duals of linear unital positive (u.p.) maps.
From the point of view of physical dynamics, one is particularly interested in unital completely positive (u.c.p.) maps.
In physics u.c.p. maps between observable algebras are referred to as \textit{channels}, and we'll use this terminology often. 
However, we do need to cover certain maps which are merely u.p., not c.p., as will be discussed in the context of detailed balance in Subsections \ref{SubsecFDB} and \ref{OndAfdSQDB}.
Therefore we define duals in a general enough form to cater for this.

The basic definition when ignoring gradings
is the following, which is the (Accardi-Cecchini) dual introduced in \cite[Proposition 3.1]{AC} 
(preceded by a special case in \cite{AH}), formulated here specifically for u.p. maps. 
\begin{definition}\label{duaalDef}
	Consider a u.p. map $E: \ca \to \cb$ such that $\nu\circ E=\m$ for a pair $\m \in \cf(\ca)$ and $\n \in \cf(\cb)$.
	The \emph{dual} 
	$$E' : \cb' \to \ca'$$
	of $E$ 
	(with respect to $\m$ and $\n$)
	is a positive linear map required to satisfy
	\begin{equation}\label{duDefEinsk}
	\left\langle \L_\m , a E'(b') \L_\m \right\rangle
	=
	\left\langle \L_\n  , E(a) b' \L_\n  \right\rangle
	\end{equation}
	for all $a\in \ca$ and $b' \in \cb'$.
\end{definition}
This dual exists (see below), and is unique
and necessarily unital
because of the cyclic and separating vectors involved in the definition. It also satisfies
$\m' \circ E' = \n'$. 
Note that the dual of $E'$ is $E'' = E$.

Recall that a linear map $E : \ca \to \cb$ is said to be \emph{even} if $$E \circ \g _\ca =\g_\cb \circ E.$$
From $\g_{\ca} (a) = g_\ca a g_\ca $ it is seen that 
$\g'_\ca  = \g_{\ca'}$, from which it easily follows that $E'$ is even if $E$ is.
In this paper our interest is in the twisted version of the dual $E'$ when $E$ is even. 
The dual $E'$ will be the special case of the twisted dual $E^{\wr}$ when using a trivial grading.

Since a von Neumann algebra and its twisted commutant are not mutually commuting, 
as opposed to $a$ and $E'(b')$ commuting in \eqref{duDefEinsk},
we initially split the definition of the twisted dual into two possibilities as follows, 
but then show that we can ultimately simplify it to Definition \ref{duaalDef}'s form.

\begin{definition}\label{verwDuaalDef}
	Consider an even u.p. map $E: \ca \to \cb$ such that $\nu\circ E=\m$ 
	for a pair $\m \in \cf_+(\ca)$ and $\n \in \cf_+(\cb)$.
	The \emph{twisted dual} 
	$$E^\wr : \cb^\wr \to \ca^\wr$$
	of $E$ 
	(w.r.t. $\m$ and $\n$)
	 is a positive linear map required to satisfy either
	\begin{equation}\label{verwrDuaal}
		\left\langle \L_\m , a E ^\wr(b^\wr) \L_\m \right\rangle
		=
		\left\langle \L_\n  , E(a) b^\wr \L_\n  \right\rangle
	\end{equation}
	for all $a\in \ca$ and $b^\wr \in \cb^\wr$, or
	\begin{equation}\label{verwrDuaal'}
		\left\langle \L_\m , E ^\wr(b^\wr) a \L_\m \right\rangle
		=
		\left\langle \L_\n  , E(a) b^\wr \L_\n  \right\rangle
	\end{equation}
	for all $a\in \ca$ and $b^\wr \in \cb^\wr$.
\end{definition}

As in the case of the dual, one sees directly from this definition that the twisted dual is 
necessarily unique, even and unital.
In addition, it satisfies
$\m^\wr  \circ E^\wr  = \n^\wr $, 
and the corresponding twisted dual of $E^\wr$ is $E^{\wr\wr} = E$.
The existence of $E^\wr$ is discussed below.

The case \eqref{verwrDuaal} can be expressed as 
$
\left\langle a \L_\m , E ^\wr(b^\wr) \L_\m \right\rangle
=
\left\langle E(a) \L_\n  , b^\wr \L_\n  \right\rangle ,
$
and \eqref{verwrDuaal'} as
$
\left\langle E ^\wr(b^\wr)^* \L_\m , a \L_\m \right\rangle
=
\left\langle E(a)^* \L_\n  , b^\wr \L_\n  \right\rangle,
$
with the latter reminiscent of the adjoint of a conjugate linear operator, giving some intuition about the two cases. 
A relevant example is the following.

\begin{exa}\label{(anti)iso}
If $E$ is a $*$-isomorphism, represented unitarily as
$
E(a) = UaU^*
$,
by elementary manipulations one has
$
E'(b') = U^*b'U
$,
and for $E$ even,
$
E^\wr(b^\wr) = U^* b^\wr U
$
in the form \eqref{verwrDuaal}.
If, instead, $E$ is a $*$-anti-isomorphism, now represented as
$
E(a) = Ua^*U^*
$
using anti-unitary $U$ (see the Appendix),
it leads to
$
E'(b') = U^*b'^*U
$,
and for the even case to
$
E^\wr(b^\wr) = U^* b^{\wr*} U
$
in the form \eqref{verwrDuaal'}.
\end{exa}

Nevertheless, the forms \eqref{verwrDuaal} and \eqref{verwrDuaal'} are equivalent in general, as seen as part of the next result.
That is, we can define $E^\wr$ solely in terms of \eqref{verwrDuaal}. 
\begin{prop}
	\label{DuaEkw}
	For $E$ as given in Definition \ref{verwDuaalDef}, the twisted dual $E^\wr$ exists if and only if the dual $E'$ does, in which case
		\begin{equation*} \label{verwrDuDef}
				E^\wr = \g _\ca^{1/2}\circ E' \circ\g_\cb^{-1/2}\,.
		\end{equation*}
	Moreover, the two possibilities \eqref{verwrDuaal} and \eqref{verwrDuaal'} in defining $E^\wr$ are in fact equivalent.
\end{prop}
\begin{proof}
	Suppose $E^\wr$ exists in the form \eqref{verwrDuaal}. We'll show that the u.p. map 
	$
	F = \g_\ca^{-1/2} \circ E^\wr \circ \g_\cb^{1/2} : \cb' \to \ca'
	$
	satisfies the requirements of $E'$.
	In terms of the notions earlier in this section as well as the notation
	$
	a_{\pm} := \frac{1}{2}( a \pm \g_\ca (a) ) \, ,
	$
	we have
	\begin{equation*}
		\begin{aligned}
			&\left\langle \L_\m , 
			a^* F(b') \L_\m \right\rangle \\
			= &
			\left\langle 
			g_\ca^{1/2} a \L_\m , E^\wr ( \g_\cb^{1/2} (b') ) \L_\m 
			\right\rangle\\
			= &
			\left\langle 
			\frac{1}{2}
			( a + a_+ - a_- - ia + ia_+ - ia_-) \L_\m , 
			E^\wr ( \g_\cb^{1/2} (b') ) \L_\m \right\rangle\\
			= &
			\left\langle 
			g_\cb^{1/2} E(a) \L_\n , 
			g_\cb^{1/2} b' \L_\n \right\rangle \\
			= &
			\left\langle \L_\n , E(a^*) b' \L_\n \right\rangle
		\end{aligned}
	\end{equation*}
	as required.
	A similar calculation, also using the fact that $\cb$ and $\cb'$ mutually commute, covers the case \eqref{verwrDuaal'}.
	That is, $E'$ indeed exists if $E^\wr$ does.
	
	Conversely, suppose $E'$ exists. Analogous calculations show that $E^\wr$ exists, being given by the u.p. map
	$\g _\ca^{1/2}\circ E' \circ\g_\cb^{-1/2}$,
	which simultaneously satisfies both \eqref{verwrDuaal} and \eqref{verwrDuaal'},
	since $\ca$ and $\ca'$ mutually commute.
	
	Combining these two directions, we also see that the two forms \eqref{verwrDuaal} and \eqref{verwrDuaal'} of $E^\wr$ are indeed equivalent.
\end{proof}

Moreover, the following can be gleaned from \cite[Proposition 11.1 and Theorem 11.2]{CDF}.
\begin{theorem} \label{twDuSys}
	The twisted dual $E^\wr$ in Definition \ref{verwDuaalDef} indeed exists.
	Furthermore, if $E$ is $n$-positive, so is $E^\wr$. If $E$ is a channel, so is $E^\wr$.
\end{theorem}

This theorem is a consequence of the work of Accardi and Cecchini \cite[Proposition 3.1]{AC}, though one can state it in different levels of generality.
See for example the earlier work \cite[Theorem 2.1]{AH} in the case of 2-positive maps, in which case the Schwarz inequality holds, leading to Hilbert space representations of the maps as contractions.
A convenient account of the Accardi-Cecchini result  is given in \cite[Theorem 2.5]{DS2}.



\begin{Rev&OT}
We'll view the existence of a (twisted) dual as a mathematically natural way to capture a general notion of reversibility of dynamics,
which is independent of whether the dual or twisted dual is considered. 
This will be discussed further in Subsections \ref{SubsecFDB} and \ref{OndAfdSQDB} in the context of detailed balance,
where the reverse of dynamics is key. 
At the same time, these duals play a cardinal role in defining transport plans between dynamical systems, as will be seen in the framework set up in the next section, as well as in subsequent sections where the theory of fermionic optimal transport between dynamical systems is systematically developed.

The fact that (twisted) duals appear in a central way in both,
emphasizes the close relation between the ideas of optimal transport between systems and the ideas of detailed balance, even though the latter is a more restricted situation. 
This suggests that optimal transport
is a natural framework in which to extend detailed balance, in particular to study some aspects of non-equilibrium by considering deviation from detailed balance in terms of optimal transport.
This is not specific to the fermionic case, but applies broadly to the quantum and classical cases, as also seen in \cite{DSS}.
\end{Rev&OT}


We record the following property, which was already implicitly used above in relation to $E'$ and $E^\wr$ being even.
\begin{prop}
	\label{compDu}
	Consider even u.p. maps $E : \ca \to \cb$ and $F : \cb \to \cc$ satisfying $\nu \circ E = \mu$ and $\xi \circ F = \nu$.
	Then
	$$(F \circ E)^{\wr}=E^{\wr} \circ F^{\wr}.$$
\end{prop}

\begin{KMS}
	As was seen in \cite{D22, D23, DSS}, the symmetry of Wasserstein distances between states and between dynamical systems is related to what is known as the KMS (or standard) dual, also known as the Petz recovery map in quantum information literature. 
	It was first defined in \cite[p. 254]{AC} (under the name ``bidual map"), but then studied for its own sake and in a broader context and different purpose in \cite{Pet, Pet88} (under the name ``dual").
	For a u.p. map 
	$E: \ca \to \cb$ such that $\nu\circ E=\m$, its \emph{KMS dual} (w.r.t. $\m$ and $\n$) 
	can be defined as the u.p. map $\cb \to \ca$
	\begin{equation}\label{KMS}
		E^\s = j_\ca \circ E' \circ j_\cb \, .
	\end{equation}
	In other words, the KMS dual is a natural way of carrying the dual over from the commutants to the algebras themselves.
	It is simple to verify that $E^\s$ is $n$-positive or completely positive, respectively, if $E$ is.
	
	\begin{rem}
		In \eqref{KMS} we use the convention that the domain of $E^\s$ is determined by the restriction imposed by that of $E'$, in this case leading to $\cb$. Similarly for compositions in the sequel.
	\end{rem}
	
	As we are aiming for fermionic Wasserstein distances, when $E$, $\m$ and $\n$ above are even,
	it is appropriate to ask
	if there is a separate fermionic version of the KMS dual in the $\bz_2$-graded case, obtained from the twisted dual instead of the dual. 
	It is natural to use 
	$$
	\g_\ca^{1/2} \circ j_\ca = j_\ca \circ \g_{\ca'}^{1/2} 
	\quad \text{and} \quad
	j_\ca \circ \g_\ca^{-1/2} = \g_{\ca'}^{-1/2} \circ j_\ca 
	$$
	instead of $j_\ca$ to go from $\ca$ to $\ca^\wr$ and back respectively (keeping in mind (\ref{Jg=gJ})). The construction corresponding to \eqref{KMS} is then
	$$ 
	j_\ca \circ \g_\ca^{-1/2} \circ E^\wr \circ \g_\cb^{1/2} \circ j_\cb
	=
	E^\s ,
	$$
	by Proposition \ref{DuaEkw},
	simply delivering the usual KMS dual of $E$. Now $E^\s$ is even.
	It is easily checked that 
	$
	( j_\ca \circ E \circ j_\cb )' = j_\cb \circ E' \circ j_\ca \, ,
	$
	hence
	\begin{equation}\label{KMSverwDu}
		\begin{aligned}
			(E^\s)^\wr 
			&= 
			\g_\cb^{1/2} \circ (j_\ca \circ E' \circ j_\cb)' \circ \g_\ca^{-1/2}
			=
			\g_\cb^{1/2} \circ j_\cb \circ E \circ j_\ca \circ \g_\ca^{-1/2}\\
			&=
			j_\cb \circ ( \g_\ca^{1/2} \circ E' \circ \g_\cb^{-1/2} )' \circ j_\ca
			=
			(E^\wr)^\s.
		\end{aligned}
	\end{equation}
	because of \eqref{Jg=gJ} and the dual of a $*$-isomorphism in Example \ref{(anti)iso}.
\end{KMS}

\section{Transport plans and dynamical systems}
\label{SecTrPl&Bal}

The starting point of optimal transport is the notion of a transport plan from one state to another. Our goal is to formulate fermionic optimal transport and Wasserstein distances in terms of fermionic or $\bz_2$-graded transport plans, also between systems. In this section we present the basic definitions of fermionic transport plans that form the foundation of our further development of fermionic optimal transport. 
In order to define different classes of fermionic transport plans relevant to us, we need conditions which will be referred to as balance conditions, also discussed in this section. 
This forms a fermionic version of transport concepts from \cite{D22, D23, DSS} and balance from \cite{DS2}.
The notation and conventions introduced here will be used throughout the sequel, along with that from the previous section.

However, we do need transport plans for the usual tensor product as well, since the Fermi case will in effect  be translated into the usual case, in terms of which the metric properties of fermionic Wasserstein distances will be proven in Section \ref{SecFW2}. 
\begin{notat}
	As we need to translate between the usual and Fermi (i.e., $\bz_2$-graded) tensor products, both in algebraic form rather than completed, 
	we'll use the symbol $\odot$ for the former, 
	and $\ftp$ for the latter, as in \cite{CDF}, including for elementary tensors. In short,
	\[
	a \odot b' \in \ca \odot \cb'
	\quad \text{and} \quad
	a \ftp b^\wr \in \ca \ftp \cb^\wr
	\]
	etc. The Fermi tensor product 
	(or $\bz_2$-graded tensor product)
	is reviewed in some detail in \cite{CDF}.
	We assume it as known.
\end{notat}

The most basic type of transport plan we will consider in the fermionic context is provided in the following definition.
Note that a \emph{state} $\om$ on a unital $*$-algebra $\cc$, is a linear functional such that $\om(1_\cc) = 1$ and $\om(c^*c) \geq 0$ for all $c \in \cc$.

\begin{definition}\label{TPforSta}
	\label{oordPlan}Consider two $\bz_2$-graded von Neumann algebras $\ca$ and $\cb$ with states 
	$\mu \in \cf_+(\ca)$ and $\nu \in \cf_+(\cb)$ 
	respectively. 
	A \emph{fermionic transport plan from }$\mu$\emph{\ to }$\nu$ is a state
	$\omega$ on the algebraic Fermi tensor product $\ca\ftp \cb^\wr$ such that
	\begin{equation}	
		\label{marginals}
		\omega(a\ftp 1_{\cb^{\wr}}) = \mu(a) 
		\text{ \ \ and \ \ }
		\omega(1_{\ca}\ftp b^\wr) = \nu^\wr(b^\wr)
	\end{equation}	
	and
	\begin{equation}
		\omega(\gamma_\ca(a)\ftp b^\wr)=\omega(a\ftp\gamma_\cb^\wr(b^\wr))
		\label{oordPl}
	\end{equation}
	for all $a\in \ca$ and $b^\wr\in \cb^\wr$. Denote the set of all fermionic
	transport plans from $\mu$ to $\nu$ by 
	$$T^{\F} (\mu,\nu).$$
\end{definition}

Condition \eqref{oordPl} is an example of a balance (or transport plan) condition which will shortly be discussed in general for systems.

We also need the case without a grading, or when the grading is simply ignored, which we state next.

\begin{definition}
	\label{oordPlan}
	Consider two von Neumann algebras $\ca$ and $\cb$ with states 
	$\mu \in \cf(\ca)$ and $\nu \in \cf(\cb)$ 
	respectively. 
	A \emph{transport plan from} $\mu$\emph{\ to} $\nu$ is a state
	$\omega$ on the algebraic tensor product $\ca\odot \cb'$ such that
	\[
	\omega(a \odot 1_{\cb'}) = \mu(a)
	\text{ \ \ and \ \ }
	\omega(1_{\ca } \odot b') = \nu'(b')
	\]
	for all $a \in \ca$ and $b' \in \cb'$. 
	We denote the set of all transport plans from $\mu$ to $\nu$ by
	$$T(\mu,\nu).$$
	When the von Neumann algebras $\ca$ and $\cb$ have $\bz_2$-gradings $\gamma_\ca$ and $\gamma_\cb$ respectively, and $\m$ and $\n$ are even, then  a \emph{graded} (or \emph{even}) transport plan from $\m$ to $\n$ is an
	$\om \in T(\mu,\nu)$
	such that 
	\[
	\omega(\gamma_\ca(a) \odot b')
	=
	\omega(a \odot \gamma_\cb'(b'))
	\]
	for all $a \in \ca$ and $b' \in \cb'$. 
	The set of graded transport plans from $\m$ to $\n$ is denoted by
	$$T^{\g}(\mu,\nu).$$
\end{definition}




Note that, even if we don't assume that $\mu$ and $\nu$ in the definition of a graded transport plan are even, as long as such a plan $\om$ exists, it follows from the balance condition that they are:
$$
\mu(\gamma_\ca(a)) 
= \omega(\gamma_\ca(a) \odot 1_{\cb'})=\omega(a \odot \g_\cb'(1_{\cb'})) 
= \om(a \odot 1_{\cb'})
= \mu(a)
$$
and similarly for $\nu$. 

Our ultimate aim is to define fermionic Wasserstein distances between fermionic quantum systems. Mathematically the systems we consider are formalized as below, with the $\bz_2$-graded case being our main interest. 
Results for the usual or non-graded case regarding metric properties of Wasserstein distances,
will nevertheless be used in Sections \ref{SecFW2} 
to arrive at corresponding results for the graded case through a translation process based on Section \ref{FTPvsTP}.

\begin{definition}\label{stelsel}
	A \emph{$\bz_2$-graded system}
	$\va = (\ca,\a,\mu,k)$ on a $\bz_2$-graded von Neumann algebra $\ca$ consists of the following:
	\emph{Coordinates} $k=(k_{1},...,k_{d})$ with $k_{1},...,k_{d}\in \ca$ for some $d\in\{1,2,3,...\}$,
	an even state $\mu\in\cf_+(\ca)$, 
	as well as even u.p. maps
	\[
	\a_\u : \ca \rightarrow \ca
	\]
	indexed by $\u$ from an arbitrary but fixed index set $\U$,
	such that 
	$$
	\mu\circ\a_\u = \mu
	$$
	for every $\u \in \U$. We write $\a = (\a_\u)_{\u \in \U}$ for these u.p. maps collectively, calling it the \emph{dynamics} of $\va$.
	
	In the case of a trivial grading, or if a $\bz_2$-grading is present but not taken into account, we simply call $\va$ a \emph{system} instead. 
\end{definition}


%

\begin{notat}\label{SupInd}
	As a short hand we'll usually suppress the index in conditions involving an index 
	$\u \in \U$.
	For example, we'll write
	$$
	\mu\circ\a = \mu 
	$$
	instead of saying that
	$
	\mu\circ\a_\u = \mu
	$
	for all $\u \in \U$.
    The notation 
    $$\vb = (\cb, \b,\nu,l)$$ 
    will be used to denote a system on $\cb$. 
    Aside from a specific case mentioned below, $d$ and the index set $\U$ can be taken as fixed throughout the paper, and in particular they are the same for all systems considered.
\end{notat}

In Definition \ref{stelsel} we used a simpler (but ultimately equivalent) indexing of the dynamics than \cite[Definition 3.1]{D23} and \cite[Definition 7.1]{DSS}. The point remains the same though: we allow $\va$ to simultaneously encompass multiple dynamics on the same von Neumann algebra with the same invariant state. The index set $\U$ could for example include a copy of $\br$ to index a one parameter group, with the rest of $\U$ similarly indexing other dynamics, including groups, semigroups or single u.p. maps.

There is a particular case, though, where we'll have no coordinates or index set indicated,
namely the system denoted as
$$(\ca,\s^\m,\m).$$
Here we use the modular group $\s^\m$ (from
Tomita-Takesaki theory) as the dynamics,
the index set effectively taken as $\br$.
Coordinates won't play any role in it. It will appear in the transport plan conditions in (\ref{FOordPl}) and (\ref{OordPl}) below, where coordinates do not yet come into play.


\begin{rem}
Physically, channels are typically of most interest as dynamics. But as will be seen in Section \ref{SecAfwVanFB}, certain merely u.p. maps related to detailed balance are also relevant. 
Moreover, complete positivity of dynamics is not required in building a theory of optimal transport for these systems.

The $k_i$'s will play the role of coordinates, analogous to the coordinate functions $(x_1,...,x_d) \mapsto x_i$ on $\br^d$, but they will only become relevant in Section \ref{SecFW2}, where we discuss fermionic Wasserstein distances.
\end{rem}

In  order to define all the transport plans relevant to us, we need to define the KMS dual of a system.
\begin{definition}
	\label{KMSduaalDef}
	The \emph{KMS dual} of the system $\va$, is the system
	$$\va^\s = (\ca, \a^\s, \mu, k),$$ 
	where
	$
	(\a^\s)_\u = \a_\u^\s := (\a_\u)^\s
	$
	for all $\u$, 
	in terms of the KMS dual (with respect to $\m$) discussed in Section \ref{duality}.
	Note that if $\va$ is $\bz_2$-graded, then so is $\va^\s$.
\end{definition}

We'll need a fermionic version of a condition which was called balance between systems in \cite{DS2}, but will here be viewed as expressing a transport plan between systems, as in \cite{DSS}. 
This condition does not involve the coordinates $k$ of a system $\va$, but only the states and dynamics.
We present this condition along with refined versions which will play a key role in fermionic Wasserstein distances between $\bz_2$-graded systems.

\begin{definition}	\label{TF(A,B)} 
	Consider two $\bz_2$-graded systems $\va$ and $\vb$.  
	A \emph{fermionic transport plan from $\va$ to} $\vb$ is an 
	$\om \in T^{\F} (\m,\n)$ such that
	\begin{equation}
		\label{FOordPl}
		\om(\a(a) \ftp b^\wr) = \om(a \ftp \b^\wr(b^\wr))
	\end{equation}
	for all $a\in \ca$ and $b^\wr\in \cb^\wr$. 
	This condition is written as
	\[
	\va \om \vb.
	\]
	We denote the set of all fermionic transport plans from $\va$ to $\vb$ by
	\[
	T^{\F}(\va,\vb) := \left\{  \om \in T^{\F}(\m,\n) : \va \om \vb \right\}  .
	\]
	Similarly the set of \emph{modular} fermionic transport plans from $\va$ to $\vb$ is
	\[
	T^{\F}_{\s}(\va,\vb) 
	:= 
	\left\{  
	\om \in T^{\F}(\va,\vb) : (\ca,\s^\m,\m) \om (\cb,\s^\n,\n) 
	\right\} .
	\]
	The set of \emph{KMS} fermionic transport plans from $\mathbf{A}$ to $\mathbf{B}$ is
	\[
	T^{\F}_{\s\s}(\va,\vb)
	:= 
	\left\{ \om \in T^{\F}_\s(\va,\vb) : \va^\s \om \vb^\s \right\}.
	\]
\end{definition}

Note that none of the sets of fermionic transport plans we've defined are empty, as they all contain $\m \ftp \n^\wr$.
The  more restricted the set of transport plans, the more properties the resulting fermionic Wasserstein distance between systems will have, as will be seen in Section \ref{SecFW2}. For example, the KMS duality requirement defining $T^{\F}_{\s\s}(\va,\vb)$, will lead to symmetry, hence delivering a pseudometric $W^{\F}_{\s\s}$ on any set of $\bz_2$-graded systems (with the same $\U$ and $d$).

As a nontrivial example, we now define an important fermionic transport plan from a state to itself, analogous to the case of usual tensor products.

\begin{definition}\label{diagFTrPl}
	The \textit{fermionic diagonal transport plan} (also called the \textit{fermionic identity transport plan}) $\d_{\nu}^{\textrm{F}}$ from $\nu \in \cf_+(\cb)$ to itself is defined as
	\[
	\d_{\nu}^{\F}(d) = \langle \L_\n , \varpi_\cb^{\F}(d) \L_\n \rangle,
	\]
	for all $d \in \cb \ftp \cb^{\wr},$ where 
	\[
	\varpi_\cb^{\F}: 
	\cb \ftp \cb^{\wr} \to B(\cg_{\cb})
	\]
	is given via
	\[
	\varpi_\cb^{\textrm{F}}(b \ftp b^{\wr})=bb^{\wr}
	\]
	for all $b \in \cb,\,b^{\wr} \in \cb^{\wr}.$
\end{definition}

\begin{remark}\label{diagIsFTrPl}
	It is elementary to show that $\d_{\nu}^{\textrm{F}}$ is a state by proving that $\varpi_\cb^{\textrm{F}}$ is a unital $*$-homomorphism. 
	Condition \eqref{marginals} is trivial to check from the definition of a diagonal transport plan, whereas \eqref{oordPl} follows from Example \ref{(anti)iso}.
	Note that in fact
	$$\d_{\n}^{\F} \in T^{\F}_{\s\s}(\vb,\vb),$$
	since the required properties are clearly satisfied.
	
	The term ``diagonal'' comes from the analogous construction in measure theory, while the term ``identity'' can be used, since we will see later that this transport plan corresponds to the identity channel (refer to Remark \ref{diag&IdCh}).
\end{remark}

In order to exploit the theory of Wasserstein distances for usual tensor products from \cite{D22, D23, DSS}, we'll express the sets of fermionic transport plans via corresponding sets of ``usual" transport plans. For the sake of clarity, we explicitly state their definition, taken from \cite{DSS}.

\begin{definition}
	\label{T(A,B)}
	Consider two systems $\va$ and $\vb$.  
	A \emph{transport plan from $\va$ to} $\vb$ is an 
	$\om \in T(\m,\n)$ such that
	\begin{equation}
		\label{OordPl}
		\om(\a(a) \odot b') = \om(a \odot \b'(b'))
	\end{equation}
	for all $a\in \ca$ and $b'\in \cb'$. 
	This condition is written as
	\[
	\va \om \vb.
	\]
	We denote the set of all transport plans from $\va$ to $\vb$ by
	\[
	T(\va,\vb) := \left\{  \om \in T(\m,\n) : \va \om \vb \right\}  .
	\]
	Similarly the set of \emph{modular} transport plans from $\va$ to $\vb$ is
	\[
	T_{\s}(\va,\vb) 
	:= 
	\left\{  \om \in T(\va,\vb) : (\ca,\s^\m,\m) \om (\cb,\s^\n,\n) \right\} .
	\]
	The set of \emph{KMS} transport plans from $\mathbf{A}$ to $\mathbf{B}$ is
	\[
	T_{\s\s}(\va,\vb)
	:= 
	\left\{ \om \in T_\s(\va,\vb) : \va^\s \om \vb^\s \right\}.
	\]
\end{definition}

In connection to $T_{\sigma}(\mathbf{A},\mathbf{B})$ we point out that
$(\sigma_{t}^{\nu})^{\prime}=\sigma_{t}^{\nu^{\prime}}$.
Note that the transport plans do not depend on the coordinates $k$ and $l$.
Furthermore, $\m \odot \n'$ is contained in all the sets of transport plans we defined.

In Subsection \ref{SecFTPvsTP} we'll show that each of the sets of fermionic transport plans defined in Definition \ref{TF(A,B)}
can be described via the corresponding set of transport plans in Definition \ref{T(A,B)}.
This will be the basis of our strategy to derive the metric properties of fermionic Wasserstein distances from those of the ``usual" quantum Wasserstein distances as developed in \cite{D22, D23, DSS}.

This will be done by simply adding the grading to the dynamics of the system in question. More formally, we can state it as follows:

\begin{definition}
	Consider a system $\va$ and a grading $\g_\ca$ of $\ca$ such that 
	$\m \circ \g_\ca = \m$. We define a system 
	$$\va^\g = (\ca, \a^\g, \m, k)$$
	from this by taking the latter's index set to be the disjoint union of $\U$ with one more element which we'll simply take as the symbol $\g$, while setting 
	$$\a^\g_\u := \a_\u \qquad\mbox{and}\qquad \a^\g_\g := \g_\ca.$$
\end{definition}

Using this notation, for $\bz_2$-graded systems we'll be able to relate $T^{\F}(\va,\vb)$ directly to $T(\va^\g,\vb^\g)$, and similarly for the other sets. The next section prepares the ground for this, by studying relevant aspects of the cyclic representations obtained from fermionic transport plans.


When the notation from this section and the previous is used in subsequent sections, all related assumptions will be in effect, 
for example when we make a statement about $T^{\F}(\va,\vb)$, all the assumptions required for it in Definition \ref{TF(A,B)} are also implied.

Having set up the basic objects and notions of interest to us, we proceed with the translation process between the usual and fermionic cases in the next two sections.

\section{Cyclic representations}
\label{SecCR}
The theory of Wasserstein distances is set in cyclic representations obtained from transport plans. 
This section presents these, along with related notation.

\subsection{Representations from fermionic transport Plans}\label{SubSecFTrPl}
Consider $\om \in T^{\textrm{F}}(\mu,\nu).$ 
We are going to define aspects of a cyclic representation for this fermionic transport plan only on an inner product space, since we haven't assumed $\om$ to be bounded.
Set 
$$
I_{\om}:=\{c \in \ca \ftp \cb^{\wr}:\om(c^*c)=0\}
\quad\text{and}\quad 
\ch_{\om}^{0}:=(\ca \ftp \cb^{\wr})/I_{\om}.
$$
In terms of the equivalence class notation 
$[c] = c+I_{\om}$,
define an inner product on $\ch_{\om}^0$ by
$$
\langle [c],[d] \rangle
:=
\om(c^* d)
$$ 
for all $c, d \in \ca \ftp \cb^{\wr}$.
Next define 
$$\pi_\om^0 : \ca \ftp \cb^{\wr} \to L(\ch_\om^0)
\quad\text{through}\quad
\pi_{\om}^{0}(c)[d]:=[cd]
$$
for all $c, d \in \ca \ftp \cb^{\wr}$,
with $L(\ch_\om^0)$ the algebra of (not necessarily bounded) linear operators
$\ch_\om^0 \to \ch_\om^0$.
Then 
\begin{equation}\label{omRepHom}
\pi_{\om}^{0}(c)\pi_{\om}^{0}(d)=\pi_{\om}^{0}(cd)
\quad\text{and}\quad
\pi_{\om}^{0}(c)\Om=[c],
\end{equation}
for all $c,d \in \ca \ftp \cb^\wr$,
in terms of 
$$
\Om := [1] = [1_\ca \ftp 1_{\cb^\wr}].
$$ 
Therefore
\begin{equation}\label{GNS}
\om(c) 
= \langle [1] , [c] \rangle
= \langle \Om , \pi_{\om}^{0}(c)\Om \rangle
\end{equation}
Let the completion of $\ch_{\om}^{0}$ be denoted by 
$
\ch_{\om} \, .
$ 
\begin{rem}\label{begrensdeBelofte}
Note that at this stage we don't need to know whether $\pi_{\om}^0(c)$ is bounded and extendable to $\ch_{\om}.$
The latter, and consequently the boundedness of $\om$, will follow from the translation process between the usual and fermionic tensor products in Section \ref{FTPvsTP}, 
as will be seen in Subsections \ref{SubsecTranslRep} to \ref{SecFTPvsTP}. 
Refer to Remark \ref{FGNSb} in particular.
Indeed, this is an illustration of the translation process.
Nevertheless,  in the proof of Proposition \ref{propFrepToUrep},
aspects of preservation of involution by $\pi_{\om}^0$ will need to be explored
directly in terms of it, rather than via the usual case.
\end{rem}

The map $\pi_{\om}^0$ induces cyclic representations associated to $\mu$ and $\nu^{\wr}$ by defining
\begin{equation*}
	\ch_\mu^0 := \pi_\omega^0(\ca \ftp 1_{\cb^{\wr}})\Omega, 
	\qquad 
	\pi_\mu^{0}(a) := \pi_\omega^0(a \ftp 1_{\cb^{\wr}})|_{\ch_\mu^0 }
\end{equation*}
and 
\begin{equation}\label{erfnu}
	\ch_\nu^0 
	:= 
	\pi_\omega^0(1_{\ca} \ftp \cb^\wr)\Omega, 
	\qquad 
	\pi_{\nu^\wr}^0(b^\wr)
	:= 
	\pi_\omega^0(1_{\ca}\ftp b^\wr)|_{\ch_\nu^0}
\end{equation}
for all $a\in \ca$ and $b^\wr\in \cb^\wr$.
In the GNS construction, $\pi_{\mu}^0(a)$ and $\pi_{\nu^{\wr}}^0(b^\wr)$ are bounded. By continuous extension, we obtain cyclic representations
\begin{equation}\label{eq:cyclicrep}
		(\ch_\m^\om , \pi_\mu^{\om}, \Omega) 
		\quad \text{and} \quad 
		(\ch_\n^\om , \pi_{\nu^\wr}^{\om}, \Omega) 
\end{equation} 
associated to $(\ca,\mu)$ and $(\cb^\wr,\nu^\wr)$ respectively,
$\ch_\m^\om$ and $\ch_\n^\om$ being the completions of $\ch_\m^0 $ and $\ch_\n^0 $.
There is a unitary equivalence $v_\n : \cg_\cb \rightarrow \ch_\nu^{\om}$ from $(\cg_\cb,\id_{\cb^\wr}, \L_\nu)$ to $(\ch_\nu^{\om},\pi_{\nu^\wr}^{\om},\Omega)$
determined  by
\begin{equation}\label{Fu}
v_\n b^\wr \L_\n = \p_{\n^\wr}^{\om}(b^{\wr}) \Om
\end{equation}
for all $b^{\wr}\in \cb^{\wr},$ leading to cyclic representations
$(\ch_\nu^{\om},\pi_\nu^{\om},\Om)$ of $(\cb,\nu)$ and $(\ch_\nu^{\om},\pi_{\nu'}^{\om},\Om)$ of $(\cb',\nu')$, given by
\begin{equation*}
		\pi_{\nu}^\om: \cb \rightarrow B(\ch_\nu^{\om}) : b \mapsto v_\n b v^*_\n
	\end{equation*}
and
\begin{equation}\label{p_n'}
	\pi_{\nu'}^\om: \cb' \rightarrow B(\ch_\nu^{\om}) : b' \mapsto v_\n b' v^*_\n.
\end{equation}

We will call $(\ch_\om^0 ,\pi_\om^0,\Om)$, along with the induced cyclic representations above, the $\om$-\textit{representation}. 
	
\subsection{Gradings and Klein isomorphisms}\label{SubsecGrad&Kl}
We have unitary representations 
$h_\mu: \ch_\m^{\om} \rightarrow \ch_\mu^{\om}$ 
and 
$h_\nu: \ch_\nu^{\om} \rightarrow \ch_\nu^{\om}$
of $\g_\ca$ and $\g_{\cb^\wr}$, in \eqref{eq:cyclicrep}, by
\begin{equation*}
	h_\mu \pi_\mu^{\om}(a) \Omega := \pi_\mu^{\om}(\gamma_\ca(a))\Omega
\quad
\text{and}
\quad
	h_\nu \pi_{\nu^\wr}^{\om}(b^\wr) \Omega := \pi_{\nu^\wr}^{\om}(\gamma_{\cb^\wr}(b^\wr))\Omega
\end{equation*}
for all $a \in \ca$ and $b^\wr \in \cb^\wr$.
Defining a unitary
$
h_\om : \ch_\om \rightarrow \ch_\om
$
through
\begin{equation}\label{h_om}
	h_\om\pi_\om^0(c) \Om := \pi_\om^0(\gamma_\ca \ftp \gamma_{\cb^\wr}(c))\Om
\end{equation}
for all $c \in \ca \ftp \cb^\wr$, we have
\begin{equation}\label{homegares}
	h_\om^2 = 1 ,
	\quad 
	h_\om|_{\ch_\mu^\om} = h_\m
	\quad \text{and} \quad 
	h_\om|_{\ch_\nu^{\om}} = h_\n \, .
\end{equation}
Since
	\begin{equation*}
		\pi_\mu^{\omega}(\gamma_\ca(a)) = 
		h_\mu \pi_\mu^{\omega}(a) h_\m \, ,
	\quad
	\pi_{\nu^\wr}^{\omega}(\gamma_{\cb^\wr}(b^\wr)) = 
	h_\nu \pi_{\nu^\wr}^{\omega}(b^\wr) h_\n
	\end{equation*}
	and 
	\begin{equation*}
	\pi_\om^{0}(\g_\ca \ftp \g_{\cb^\wr}(c)) = h_\om \pi_\om^0(c) h_\om
	\end{equation*}
for all $a \in \ca$, $b^\wr \in \cb^\wr$ and $c \in \ca \ftp \cb^\wr$,
we can define 
\begin{equation}
\eta_\mu \in \Aut(B(\ch_\mu^{\om})),
\quad 
\eta_{\nu^\wr} \in \Aut(B(\ch_\nu^{\om})), 
\quad 
\eta_\om : L(\ch_\om^0) \to L(\ch_\om^0)\notag
\end{equation}
by
\begin{equation*}\label{eq:etah}
	\eta_\mu(r):= h_\mu r h_\mu, 
	\quad 
	\eta_{\nu^\wr}(s):= h_\nu s h_\nu, 
	\quad 
	\eta_\omega(t):= h_\omega t h_\omega 
\end{equation*}
for all $r \in B(\ch_\mu)$,  $s \in B(\ch_\n)$ and $t \in L(\ch_\om^0)$.
Then
\begin{equation}\label{eq:etapigamma}
	\eta_\mu \circ \pi_\mu^{\om} = \pi_\mu^{\om} \circ  \gamma_\ca \, , \quad 
	\eta_{\nu^\wr} \circ \pi_{\nu^\wr}^{\om} = \pi_{\nu^\wr}^{\om} \circ  \gamma_{\cb^\wr}
\end{equation}
 and
\begin{equation}\label{etaCompi}
\eta_\omega \circ \pi_\omega^0 = 
\pi_\omega^0 \circ ( \g_\ca \ftp \g_{\cb^\wr}) .
\end{equation}

\begin{remark}
It is worth noting that, while $\m$ and $\n^\wr$ are faithful, $\om$ in general is not. Hence, $\p_\m^\omega$ and $\p_{\n^\wr}^\omega$ are faithful representations, though $\p_\om^0$ need not be. 
So $\eta_\m$ and $\eta_{\n^\wr}$ are essentially copies of $\gamma_\ca$ and $\gamma_{\cb^\wr}$, just in the $\omega$-representation, however, the same need not be true of $\eta_\om$ versus $\g_\ca \ftp \g_{\cb^\wr}$. 
But $\eta_\om$ is a natural part of our theory. 
\end{remark}
Combining the identity 
$\pi_{\nu^\wr}^\om(\cb^\wr) = v_\n \cb^\wr v^*_\n$ 
with the second item in \eqref{eq:etapigamma}, we obtain 
		\begin{equation}\label{h_nuVsg_nu}
			h_\nu v_\n = v_\n g_\cb \, .
		\end{equation}
Moreover, defining $\eta_\nu(s):= h_\nu s h_\nu$ for all $s\in \pi_\nu^\om(\cb),$ it follows that
		\begin{equation*}
			\eta_\nu \circ \pi_\nu^{\om}= \pi_\nu^{\om} \circ \gamma_\cb \, .
		\end{equation*}
Then $\eta_\nu$ is a $\mathbb{Z}_2$-grading of $\pi_\nu^{\om}(\cb)$, namely $\g_\cb$ in the $\om$-representation.


Following Section \ref{duality}, but now in the $\omega$-representation, we set
\begin{equation}\label{kleinexpr}
	q^{\pm}_\nu := \frac{1}{2}(1 \pm h_\nu) 
	\quad \text{and} \quad 
	h_\nu^{1/2}:= q^+_\nu -iq^-_\nu.
\end{equation}
The \emph{Klein automorphism} corresponding to the grading $\eta_\nu$ of $\pi_\nu^{\om}(\cb)$  is defined as 
\begin{equation*}
	\eta^{1/2}_\nu : B(\ch_\nu^{\om}) \rightarrow B(\ch^{\om}_\nu)
	                          : s \mapsto h_\nu^{1/2} s {h_{\nu}^{-1/2}},
\end{equation*}
with its inverse  denoted by $\eta^{-1/2}_\n$.
From \eqref{kleinexpr} and \eqref{h_nuVsg_nu} it follows that
	\begin{equation}\label{KleinRepIntertw}
		h^{1/2}_\nu v_\n = v_\n g^{1/2}_\cb,
	\end{equation}
	hence
	\begin{equation}
    \label{KleinIntertw}
		 \eta^{1/2}_\n \circ \pi_{\n'}^\om
		 = 
		 \pi_{\n^\wr}^\om \circ \g^{1/2}_\cb .
	\end{equation}

Setting
\begin{equation}\label{GNS-Klein}
q^{\pm}_\om := \frac{1}{2}(1 \pm h_\om) 
\quad \text{and} \quad 
h^{1/2}_\om := q^+_\om -iq^-_\om ,
\end{equation}
allows us to define the Klein automorphism of $L(\ch_\om^0)$ corresponding to $\eta_\om$, as
\begin{equation*}
	\eta^{1/2}_\om : 
	L(\ch_\om^0) \rightarrow L(\ch_\om^0) :  
	t \mapsto  h^{1/2}_\om t h^{-1/2}_\om,
\end{equation*}
with its inverse  denoted by $\eta^{-1/2}_\omega$.

The identity \eqref{KleinIntertw} along with $\eta^{1/2}_\om$ will be of use to translate representations from Fermi to usual tensor products in Subsection \ref{SubsecTranslRep}
(see the proofs of Lemma \ref{Vertaalde_u} and Proposition \ref{propFrepToUrep} respectively).

\subsection{Representations from usual transport plans}\label{SubsecTrivgrad}
We also need the case of the usual tensor product in Section \ref{FTPvsTP}, since there structures on the Fermi tensor product are translated into corresponding structures on the usual tensor product. Therefore, now consider 
\[
\om \in T(\m,\n).
\]
In this special case of Subsection \ref{SubSecFTrPl} (for a trivial grading), $\pi_{\om}^0(c)$ is bounded for all $c \in \ca \odot \cb^{\prime}$ hence it is continuously extendable to a cyclic representation $\pi_{\om}$ on $\ch_\om$ in the usual bounded operator sense, since $\om$ can be extended to a 
state on the $C^*$-algebra $\ca \otimes_{\text{max}} \cb'$. 


The same notation as in Subsection \ref{SubSecFTrPl} will be used for the $\om$-representation, except that to avoid any confusion later, the unitary equivalence \eqref{Fu} will rather be written as
\begin{equation}
	\label{GnuHnu}
	u_\n :\cg_\cb\rightarrow \ch_\n^\om \, ,
\end{equation}
being given by
\begin{equation}
\label{v_n}
u_\n b' \L_\n:=\pi_{\nu'}^\om (b')\Omega
\end{equation}
for all $b^{\prime}\in \cb^{\prime}$.
Then
\begin{equation*}
\label{pi_nu'}
\pi_{\n'}^{\om}(b')=u_\n b' u_\n^*
\end{equation*}
for all $b'\in \cb'$.
By setting
\begin{equation*}\label{pi_nu}
\pi_\n^\om (b) := u_\n b u_\n^*
\end{equation*}
for all $b\in \cb$, we also have a cyclic representation $(\ch_{\n}^{\om},\pi_{\nu}^{\om},\Om)$ of $(\cb,\n)$. 

These last two representations can also be connected in the following way:
Denote the modular conjugation for $\cb$ associated to $\Lambda_{\nu}$ by
$J_{\cb}$ and the modular conjugation for $\pi_{\n}^\om(\cb)$ associated to $\Omega$ by $J_{\cb}^\om$.
Then one finds that
\[
J_{\cb}^\om u_\n = u_\n J_{\cb}
\]
from which it follows that
\[
\pi_{\nu}^{\omega}=j_{\cb}^{\omega}\circ\pi_{\nu^{\prime}}^{\omega}\circ
j_{\cb}%
\]
where $j_{\cb}^{\omega}:=J_{\cb}^{\omega}(\cdot)^{\ast}J_{\cb}^{\omega}$ and
$j_{\cb}:=J_{\cb}(\cdot)^{\ast}J_{\cb}$ on $B(\ch_{\nu}^{\omega})$ and
$B(\cg_{\cb})$ respectively. 


Refining this setup, assume that
\[
\om \in T^\g(\m,\n)
\]
for gradings $\g_\ca$ and $\g_\cb$. We can define $h_\om$, $h_\om^{1/2}$, $h_\n$, etc., analogously to Subsection \ref{SubsecGrad&Kl}, but in terms of the usual tensor product. One then again finds for example that
\begin{equation}\label{KleinUsRepIntertw}
	h^{1/2}_\n u_\n = u_\n g^{1/2}_\cb,
\end{equation}
as in \eqref{KleinRepIntertw}.

\section{Fermionic versus usual transport plans}
\label{FTPvsTP}
In this section we establish a one-to-one correspondence between fermionic and usual transport plans.
This is intertwined with the representation of fermionic transport
plans in terms of even channels.  
These two related ideas will form core components in our approach to fermionic
Wasserstein distances in the next section, translating the fermionic to the usual case and back. 
They in turn build on the representations from the previous section.

\subsection{Supercommutation}\label{SubsecSupCom}
The term ``supercommutation'' 
refers to a substitute for the commutation between $\ca\odot 1_{\cb'}$ and $1_\ca \odot \cb'$ in the case of $\ca\ftp 1_{\cb^\wr}$ and $1_\ca \ftp \cb^\wr$.
For our purposes it is needed in an $\om$-representation from the fermionic transport plan being considered.
On the one hand, the results of Subsection \ref{SubsecTranslRep} uses this to obtain a representation of the corresponding usual tensor product.
On the other hand we subsequently also essentially reverse supercommutation, 
with a representation for the Fermi tensor product obtained from that of the usual tensor product. 
These are complementary results of central importance to our approach to fermionic Wasserstein distances, since they lead to the one-to-one correspondence between 
$T^{\F}(\m , \n)$ and $T^\g(\m ,\n )$, 
as well as a channel representation of fermionic transport plans. 

Our basic supercommutation result is the next simple lemma. It builds directly on Subsections \ref{SubSecFTrPl} and \ref{SubsecGrad&Kl}'s notation and arguments. 

\begin{lem}
\label{eq39}
Given any $\om \in T^\text{\emph{F}}(\mu,\nu)$, it follows that for all
$a \in \ca$ and $b^{\wr} \in \cb^{\wr}$ we have
\begin{equation}\label{eq:com}
	\pi_\omega^0(a \ftp 1) \eta_{\omega}^{\pm1/2}\circ\pi_\omega^0(1 \ftp b^\wr) = 
	\eta^{\pm1/2}_\omega \circ \pi_\omega^0(1  \ftp b^\wr) \pi_\omega^0(a \ftp 1).
\end{equation}
\end{lem}
\begin{proof}
First we recover a formula for Klein isomorphisms which \cite[Section 5]{CDF} used as the definition (and called a twist automorphism), specifically for $\eta^{1/2}_\omega$ from Subsection \ref{SubsecGrad&Kl} in the $\om$-representation, 
allowing for unbounded operators
(not considered in \cite{CDF}).
In this context, Subsection \ref{SubsecGrad&Kl} defined $\eta_{\om}$ and $\eta_{\om}^{1/2}$ as maps from $L(\ch_\omega^0)$ to itself.
We also define
\begin{equation*}
	t_{\pm} :=  \frac{1}{2} \left( t \pm \eta_\om(t) \right) 
\end{equation*}
for all $t \in L(\ch_\om^0)$.
Simple calculations show that
$h_\om t_\pm = \pm t_\pm h_\om$, thus
	\begin{equation*}
		t_+ q^{\pm}_\omega = q^{\pm}_\omega t_+ 
		\quad 
		\text{and} 
		\quad t_{-} q^{\pm}_\omega = q^{\mp}_\omega t_{-}
	\end{equation*}
	in terms of \eqref{GNS-Klein}, and consequently
	\begin{equation}\label{Kleinpm}
		\eta^{\pm1/2}_\om (t) 
		=
		t_+ \pm ih_\om t_{-} \, . 
	\end{equation}
Writing 
$
v:=\pi_{\om}^0(a \ftp 1_{\cb^{\wr}})
$
and
$
w:=\pi_{\om}^0(1_{\ca} \ftp b^{\wr}),
$
it is seen from \eqref{etaCompi} that 
$$v_+=\pi_{\om}^0(a_+ \ftp 1_{\cb^{\wr}}),\,\,\,\,\,\,v_-=\pi_{\om}^0(a_{-} \ftp 1_{\cb^{\wr}})$$
and
$$w_+=\pi_{\om}^0(1_{\ca} \ftp (b^{\wr})_+),\,\,\,\,\,\,w_-=\pi_{\om}^0(1_{\ca} \ftp (b^{\wr})_-).$$
Combining this with \eqref{omRepHom} and \eqref{Kleinpm}, the result is obtained.
\end{proof}

\subsection{Translating representations}\label{SubsecTranslRep}
This subsection sets the stage for Section \ref{SecFW2}, by showing how to translate between representations associated to fermionic and usual transport plans respectively. 

\medskip
\noindent\textbf{From Fermi to usual tensor products.} 
Consider any 
$$\om \in T^\text{F}(\mu,\nu).$$
We aim to define a representation of $\ca \odot \cb'$.
Begin by defining 
$$
\pi_\om^{\Us} : \ca\odot\cb' \to L(\ch_\omega^0)
$$
in the $\om$-representation from Subsection \ref{SubSecFTrPl} through
\begin{equation}\label{fermRepTrans}
	\pi_\om^{\Us} (a \odot b')
	=
	\pi_\om^0(a \ftp 1_{\cb^\wr}) 
	\eta_\om^{-1/2} \circ \pi_\om^0(1_\ca \ftp \g^{1/2}_{\cb}(b')) 
\end{equation}
for all $a \in \ca,\, b' \in \cb'$.
The superscript $\Us$ indicates ``usual," as opposed to fermionic.
Then we have the following result:
\begin{proposition}\label{propFrepToUrep}
	The operator $\pi_\om^{\Us} (c)$ defined above is bounded for every 
	$c \in \ca\odot\cb'$, hence extends continuously and uniquely to an element 
	of $B(\ch_\om)$. 
	This provides us with a cyclic representation 
	$(\ch_\om , \pi_\om^{\Us} , \Om)$ 
	associated to the state $\om_{\odot}$ on $\ca \odot \cb'$ defined by 
	$$
	\om_{\odot}(c)
	=
	\langle \Om, \pi_\om^{\Us} (c) \Om \rangle,
	$$
	for all $c \in \ca\odot\cb'$.
	I.e., $\p_{\om_\odot} = \p_\om^{\Us}$, 
	in terms of Subsection \ref{SubsecTrivgrad}. 
	In addition, 
	$$\om_\odot \in T^\g(\m , \n).$$
\end{proposition}
\begin{proof}
It is easily seen that 
$\pi_\om^{\Us}(a \odot b')\ch_{\om}^0 \subset \ch_{\om}^0$ for all $a \in \ca, b' \in \cb'$. 
It follows 
by \eqref{omRepHom} and Lemma \ref{eq39} that
$$
\pi_\om^{\Us}(cd)=\pi_\om^{\Us}(c)\pi_\om^{\Us}(d)
$$
for all $c, d \in \ca \odot \cb'$.
Next, we need to find an appropriate version of preservation of involution. 
It is straightforward to show from the basic definitions, 
that
$
\pi_{\om}^0(a \ftp 1_{\cb^{\wr}})^*[c \ftp d^{\wr}]
=
\pi_{\om}^0(a^* \ftp 1_{\cb^{\wr}})[c \ftp d^{\wr}]
$
for all $a,c \in \ca, d^\wr \in \cb^\wr.$
In particular it follows that
$\pi_{\om}^0(a \ftp 1_{\cb^{\wr}})^*\ch_{\om}^0 
\subset 
\ch_{\om}^0$.
Similarly 
$\pi_{\om}^0(1_{\ca}\ftp b^{\wr})^*\ch_{\om}^0 
\subset 
\ch_{\om}^0$.
Hence
$$\pi_{\om}^0(a \ftp 1_{\cb^{\wr}})^*|_{\ch_{\om}^0}
=
\pi_{\om}^0(a^* \ftp 1_{\cb^{\wr}})
\text{ and }
\pi_{\om}^0(1_{\ca} \ftp b^{\wr})^*|_{\ch_{\om}^0}
=
\pi_{\om}^0(1_{\ca} \ftp (b^{\wr})^*).$$
From the above, along with Lemma \ref{eq39} and the general operator relation $T^*S^* \subset (ST)^*,$ we have 
\begin{equation*}
	\begin{aligned}
		\pi_\om^{\Us}((a \odot b')^*)
		&=\pi_{\om}^0(a^* \ftp 1_{\cb^{\wr}})
		\eta_{\om}^{-1/2} \circ \pi_{\om}^0(1_{\ca} \ftp {\g_{\cb}}^{1/2}((b')^*))\\
		&=\pi_\om^{\Us}(a \odot b')^*|_{\ch_{\om}^0} \, ,
	\end{aligned}
\end{equation*}
to conclude that $$\pi_\om^{\Us}(c^*)=\pi_\om^{\Us}(c)^*|_{\ch_{\om}^0}$$
for all $c \in \ca \odot \cb'$.
Thus $\om_\odot$ is indeed a state on $\ca \odot \cb'$.

Next we verify that $\Om$ is cyclic for $\pi_\om^{\Us}$. 
For any $a \in \ca$ and $b' \in \cb'$, from the definition of $\pi_\om^{\Us}$ and \eqref{KleinIntertw}
we have 
$$
	\pi_\om^{\Us}(a \odot b')\Om=\pi_{\om}^0(a \ftp 1_{\cb^{\wr}})\pi_{\nu'}^{\om}(b')\Om.
$$
In addition, 
\begin{align}
	\pi _{\nu ^{\wr }}^{\om}(\cb^{\wr })\Omega 
	&= v_\n \cb^{\wr}\Lambda _{\nu }
	= v_\n \g_{\mathcal{B}}^{1/2}(\cb')\Lambda _{\nu }
	= v_\n g_{\cb}^{1/2}\cb'\Lambda_{\nu }
	= v_\n \cb'\Lambda _{\nu }
	\notag\\
	&=\pi _{\nu'}^{\om}(\cb')\Omega,
	\notag
\end{align}
with the second to last equality following from the definition of 
$g_\cb^{\pm 1/2}$ as 
$\frac{1}{2}[( 1+g_\cb )\mp i( 1-g_\cb )]$ 
and the fact that 
$g_{\cb}\cb'\Lambda _{\nu }=\cb'\Lambda _{\nu }$.
Hence
\begin{equation*}
	\begin{aligned}
		\pi_\om^{\Us}(a \odot \cb')\Om
		&=\pi_{\om}^0(a \ftp 1_{\cb^{\wr}})\pi_{\nu^{\wr}}^{\om}(\cb^{\wr})\Om 
		=\pi_{\om}^0(a \ftp 1_{\cb^{\wr}})\pi_{\om}^0(1_{\ca}\ftp \cb^{\wr})\Om\\
		&=\pi_{\om}^0(a \ftp \cb^{\wr})\Om.
	\end{aligned}
\end{equation*}
As a result 
$
\pi_\om^{\Us}(\ca \odot \cb')\Om=\ch_{\om}^0
$ 
is indeed dense in $\ch_{\om}$. 


As the state $\om_{\odot}$ is defined on the usual tensor product, it is bounded in the maximal $C^*$-norm. 
By the usual GNS arguments, it follows that $\pi_\om^{\Us}(c)$ is bounded and has a unique continuous extension $\pi_\om^{\Us}(c) \in B(\ch_{\om})$,
and that
$\pi_\om^{\Us} : \ca \odot \cb' \to B(\ch_{\om})$
is a $*$-homomorphism. 
Hence,
$(\ch_{\om}, \pi_\om^{\Us} , \Om)$ 
is a cyclic representation of $(\ca \odot \cb',\om_{\odot})$.


We proceed to show that 
$\om_\odot \in T^\g(\m , \n).$ 
First,
$
\p_\om^{\Us} (a \odot 1_{\cb'}) |_{\ch_\om^0}
=
\p_\om^0 (a \ftp 1_{\cb^\wr})
$,
hence by \eqref{GNS},
\[
\om_\odot (a \odot 1_{\cb'})
=
\om(a \ftp 1_{\cb^\wr})
=
\m(a).
\]
On the other hand, by \eqref{KleinIntertw},
$
\p_\om^{\Us} (1_\ca \odot b') \Om
=
\pi_{\n'}^\om(b') \Om,
$
thus 
\[
\om_\odot (1_\ca \odot b') 
= 
\left\langle \Om , \pi_{\n'}^\om(b') \Om \right\rangle
=
\n'(b').
\]
Lastly, 
from \eqref{fermRepTrans} and \eqref{etaCompi} 
(and Example \ref{(anti)iso}),
\begin{equation*}
	\begin{aligned}
		\p_\om^{\Us} (\g_\ca(a) \odot b') \Om
		&=
		h_\om \pi_\om^0(a \ftp 1_{\cb^{\wr}}) h_\om
		h_\om^{-1/2} \pi_\om^0(1_\ca \ftp \g^{1/2}_{\cb}(b')) \Om \\
		&=
		h_\om \pi_\om^0(a \ftp 1_{\cb^{\wr}}) 
		h_\om^{-1/2} \pi_\om^0(1_\ca \ftp \g^{1/2}_{\cb}(\g_{\cb'}(b')))  \Om \\ 
		&=
		h_\om \p_\om^{\Us} (a \odot \g'_\cb(b')) \Om ,
	\end{aligned}
\end{equation*}
hence
$
\om_\odot (\g_\ca(a) \odot b') 
= 
\left\langle h_\om\Om , \p_\om^{\Us} ((a \odot \g'_\cb(b')\Om) \right\rangle
=
\om_\odot (a \odot \g_\cb'(b'))
$,
as required.
\end{proof}


In the proof of Theorem \ref{Thm1-to-1}, which is what this section is working toward, we need to express certain objects obtained from the cyclic representation associated to 
$\om_\odot$ 
in terms of the corresponding objects in the $\om$-representation of Subsection \ref{SubSecFTrPl}.
The lemma below,
in terms of the unitary equivalence (\ref{GnuHnu}), 
now for $\om_\odot$,
$$u_\n : \cg_\cb \to \ch_\n^{\om_\odot}\, ,$$ 
will play a role in that.
\begin{lem}	\label{Vertaalde_u}
	Given the above, it follows in terms of (\ref{Fu}) that
	\begin{equation*}
		u_\n = v_\n \, .
	\end{equation*}
\end{lem}
\begin{proof}
By \eqref{v_n}, \eqref{erfnu}, \eqref{GNS-Klein}, \eqref{homegares}, \eqref{kleinexpr}, \eqref{KleinIntertw} and \eqref{p_n'},
$$
u_\n  b'\L_\n 
= 
\p_{\om_\odot}(1 \odot b') \Om 
= 
\eta_\om^{-1/2} \circ \pi_\om^0(1_\ca \ftp \g^{1/2}_{\cb}(b')) \Om
=
\p^\om_{\n'}(b') \Om
=
v_\n b' \L_\n,
$$
for all $b' \in \cb'$. 
\end{proof}

\medskip

\noindent\textbf{From usual to Fermi tensor products.} 
Conversely, consider any 
$$\om \in T^\g (\mu,\nu),$$
giving $\pi_\om$ and $h_\om$ as in Subsection \ref{SubsecTrivgrad},
in terms of the usual tensor product. 
This leads to the $\bz_2$-grading $\eta_{\om}$
on $B(\ch_\om)$, 
and the corresponding Klein automorphism $\eta _\om^{1/2}$,
as in Subsection \ref{SubsecGrad&Kl}.


\begin{proposition}\label{propUrepToFrep}
	The map $\pi _{\omega }^{\text{\emph{F}}}$ from  $\mathcal{A}\ftp\mathcal{B}^{\wr }$ to $B(\ch_\om)$ obtained from $\pi _{\omega }$ via
	\begin{equation}\label{UrepToFrep}
		\pi _{\omega }^{\text{\emph{F}}}(a \ftp b^{\wr })
		:=
		\pi _{\omega }(a\odot 1_{\cb'})
		\eta _{\omega }^{1/2}\circ \pi _{\omega}(1_{\ca}\odot \gamma _{\cb}^{-1/2}(b^\wr ))
	\end{equation}
	for all $a\in \mathcal{A}$ and $b^{\wr }\in \mathcal{B}^{\wr }$, gives a cyclic representation $(\ch_{\omega},\pi _{\omega }^{\text{\emph{F}}},\Om)$ associated to the state $\om_{\ftps}$ on $\ca\ftp\cb^{\wr }$ defined by	
	\begin{equation} \label{FermTPfromUs}
		\om_{\ftps}  = 
		\left\langle 
		\Om , \pi _{\om}^{\F}(\cdot )\Om
		\right\rangle .
	\end{equation}
\end{proposition}
\begin{proof}
As in \eqref{Kleinpm}, we have
\begin{equation}
	\eta_\om^{\pm 1/2}(t)=t_{+}\pm ih_{\om}t_{-} \, .  \label{KleinFormAgain}
\end{equation}
for all $t \in B(\ch_\om)$. Writing 
\begin{equation}\label{ef}
	e:=\pi _{\om}(a\odot 1_{\mathcal{B}'})\text{ \ \ and \ \ }
	f:=\pi _{\omega }(1_{\mathcal{A}}\odot b')
\end{equation}
as well as 
\begin{equation*}
	l:=e\eta _{\omega }^{1/2}(f)\text{ \ \ and \ \ }m:=\eta _{\omega }^{1/2}(f)e \, ,
\end{equation*}%
we find by simple calculation, because of the usual tensor product in \eqref{ef},
and using $e_{-}h_{\omega}=-h_{\omega}e_{-}$, that
\begin{equation}
	l=m-2ih_{\omega }e_{-}f_{-}  \label{comDev}
\end{equation}
for any $a\in \mathcal{A}$ and $b'\in \mathcal{B}'$. In
particular, if either $e_{-}=0$ or $f_{-}=0$, we have $l=m$.
This simple identity allows us to define a representation of $\mathcal{A}\ftp\mathcal{B}^{\wr }$ via \eqref{UrepToFrep}.
To prove that $\pi _{\omega }^{\text{F}}$ is
indeed a $\ast $-homomorphism, we use (\ref{comDev})
keeping in mind that 
$\gamma _{\mathcal{B}}^{-1/2}(b^{\wr })_{\pm }=
\gamma _{\mathcal{B}}^{-1/2}(b_{\pm }^{\wr })$.
Simply calculate:
\begin{equation*}
\begin{aligned}
\pi _{\om}^{\text{F}}((a\ftp b_{-}^{\wr })&(c_{-}\ftp d^{\wr })) 
	=-\pi _{\om}^{\text{F}}((ac_{-}) \ftp (b_{-}^{\wr }d^{\wr }))\\
	=&-\pi _{\om}((ac_{-}) \odot 1_{\cb'})
	\eta_{\om}^{1/2}\circ 
	\pi _{\om}(1_{\ca} \odot (\g _{\cb}^{-1/2}(b^{\wr })_{-}\g _{\cb}^{-1/2}(d^{\wr }))) \\
	=&-\pi _{\om}(a\odot 1_{\cb'}) 
	\times \\
	&[\eta _{\om}^{1/2}\circ \pi _{\om}(1_{\ca}\odot (\g_{\cb}^{-1/2}(b^{\wr })_{-})
	\pi _{\om }(c_{-} \odot 1_{\cb'}) 
	\notag\\
	&-2ih_{\om}\pi _{\om}(c_{-}\odot 1_{\cb'})
	\pi_{\om}(1_{\ca}\odot \g_{\cb}^{-1/2}(b^{\wr})_{-})] 
	\times  
	\notag\\
	&\eta _{\om}^{1/2}\circ \pi _{\om}(1_{\ca}\odot \g_{\cb}^{-1/2}(d^{\wr })) 
		\notag\\
	=&-\pi _{\om}^{\text{F}}(a\ftp b_{-}^{\wr })\pi _{\om}^{\text{F}}(c_{-}\ftp d^{\wr })+r \, .
	\notag
\end{aligned}
\end{equation*}
Here $r$ is defined to be
\begin{align}
	r =&\  2i\pi _{\om}(a\odot 1_{\cb'})h_{\om}\pi
	_{\om}(1_{\ca}\odot \g_{\cb}^{-1/2}(b^{\wr})_{-})
	\times  
	\notag\\
	&\pi _{\om}(c_{-}\odot 1_{\cb'})
	\eta _{\om}^{1/2}\circ \pi _{\om}(1_{\ca}\odot \g _{\cb	}^{-1/2}(d^{\wr })) 
	\notag\\
	=&\ 2\pi _{\om}(a\odot 1_{\cb'})
	\eta _{\om}^{1/2}\circ \pi _{\om}(1_{\ca}\odot \g_{\cb}^{-1/2}(b_{-}^{\wr }))
	\pi_{\om }^{\text{F}}(c_{-}\ftp d^{\wr }) 
	\notag\\
	=&\ 2\pi _{\om}^{\text{F}}(a\ftp b_{-}^{\wr })\pi _{\om}^{\text{F}}(c_{-}\ftp d^{\wr }).
	\notag
\end{align}
where we could use (\ref{KleinFormAgain}) for the second equality, 
as it is easily checked that
\begin{equation*}
\begin{aligned}
\eta _{\om}(\pi _{\om}(1_{\ca}\odot \g_{\cb	}^{-1/2}(b^{\wr })_{-}))
	= &\pi _{\om}(\g_{\ca}(1_{\ca}) \odot \g_{\cb'} (\g_{\cb}^{-1/2}(b^{\wr })_{-})) \\
	=&-\pi _{\om}(1_{\ca}\odot \g_{\cb}^{-1/2}(b^{\wr})_{-}),
	\notag
\end{aligned}
\end{equation*}
that is, 
\begin{equation*}
	\pi _{\om}(1_{\ca}\odot \g_{\cb}^{-1/2}(b^{\wr})_{-})
	=
	\pi _{\om}(1_{\ca}\odot \g _{\cb	}^{-1/2}(b^{\wr })_{-})_{-} \, ,
\end{equation*}
thus only the second term from (\ref{KleinFormAgain}) is involved. Hence%
\begin{equation*}
	\pi _{\om}^{\text{F}}((a\ftp b_{-}^{\wr })(c_{-}\ftp d^{\wr }))
	=
	\pi _{\om}^{\text{F}}(a\ftp b_{-}^{\wr })\pi _{\om}^{\text{F}}(c_{-}\ftp d^{\wr }).
\end{equation*}
Similarly (but simpler as one then uses (\ref{comDev}) in the form $l=m$)
for the cases $b_{+}^{\wr },c_{-}$, etc., to show that 
$\pi _\om^{\F}$ is a homomorphism. 
Analogously for preservation of involution to show that it is a $\ast $-homomorphism. 
Therefore
\begin{equation*}
	\pi _\om^{\text{F}}:\mathcal{A}\ftp\mathcal{B}^{\wr }\rightarrow 
	B(\mathcal{H}_{\omega })
\end{equation*}
is indeed a representation.


Next we show that it is cyclic. Note that 
\begin{equation}\label{p_omBep}
	\eta_\om^{1/2} \circ \p_\om(1_\ca \odot b')|_{\ch^\om_\n}
	=
	h_\n^{1/2}\pi _{\n'}^{\om}(b')h_{\nu }^{-1/2}
	=
	\p_{\n^{\wr }} (\g_\cb^{1/2}(b')),
\end{equation}
with 
$
\p_{\n^\wr}  := u_{\nu }(\cdot )u_{\nu }^*
$
serving as a cyclic representation of $\n^\wr$,
since $h_\n^{1/2}u_\n = u_\n g_\cb^{1/2}$ according to (\ref{KleinUsRepIntertw}).
Consequently, 
\begin{align}
	\p_{\om}^{\F}(\ca\ftp\cb^{\wr })\Om
	&=\pi_{\om}(\ca\odot 1_{\cb'})
	\eta _{\om}^{1/2}\circ \pi _{\om }(1_{\ca}\odot \g_{\cb}^{-1/2}(\cb^{\wr }))
	|_{\ch^\om_\n}\Om
	\notag \\
	&=\pi _{\om}(\ca\odot 1_{\cb'})\pi_{\nu^{\wr}} (\cb^{\wr })\Om  
	\notag\\
	&=\pi _{\om}(\ca\odot 1_{\cb'})\pi_{\nu'}^{\om}(\cb')\Om
	\notag\\
	&=\pi _{\om}(\ca\odot \cb')\Om
	\notag
\end{align}
which is indeed dense in $\mathcal{H}_{\omega }$ as $\pi _{\omega }$ is
cyclic. Note that here we've used the fact that 
$
	\pi_{\n^\wr}(\cb^\wr )\Om
	=
	\pi _{\nu'}^{\om}(\cb')\Omega,
$
as in the proof of Proposition \ref{propFrepToUrep}.
It is of course clear that $\omega _{\ftps}$
is a state on $\ca \ftp \cb^\wr$, as $\p_{\om}^{\F}$ is a representation, and indeed
$\p_{\om}^{\F}(1_\ca \ftp b^{\wr })|_{\ch^\om_\n} = \p_{\n^\wr} (b^\wr)$
from \eqref{p_omBep}, in step with Section \ref{SubSecFTrPl}'s conventions for an $\om_{\ftps}$-representation in \eqref{erfnu}.
\end{proof}



\subsection{Fermionic transport plans as channels}\label{SubsecEvsom} 
We take the first steps to show that fermionic transport plans can be expressed through even channels, in analogy to the usual case. In this subsection we restrict ourselves to the fermionic transport plans $\om_{\ftps}$ for $\om\in T^\g(\m , \n)$, but the next subsection will confirm that this in fact covers all fermionic transport plans.


Recall from \cite[Section 3]{DS2} that for $\om \in T(\m , \n)$ there is a uniquely determined function $E_\om : \ca \to \cb$  such that
\begin{equation}\label{omUitEom}
	\om = \d_\nu \circ (E_\om\odot \id_{\cb'}).
\end{equation}
Here $\d_\n$ is the (usual) \textit{identity (or diagonal) transport plan}  from $\n$ to itself,
obtained as the case of $\d^{\F}_\n$ in Definition \ref{diagFTrPl} when the grading is trivial.
This function $E_\om$ is a (necessarily normal, i.e., $\s$-weakly continuous) channel such that 
$\n \circ E_\om = \m$. 
It can be also be defined by the formula
\begin{equation}\label{Eformule}
	E_\om = u_\n^* P_\n \pi _\om(a\odot 1_{\cb'})u_\n
\end{equation}
in terms of the projection $P_\n \in B(\ch_\om)$ of $\ch_\om$ onto $\ch_{\nu}^{\om}$, using the notation of Subsection \ref{SubsecTrivgrad}
(\textit{cf.} \cite[Proposition 3.1 and Theorem 3.2]{DS2}).
Conversely, given a channel $E : \ca \to \cb$ such that 
$\n \circ E = \m$, we have 
$\d_\nu \circ (E \odot \id_{\cb'}) \in T(\m , \n)$.

The dual of $E_\om$ corresponds to the \emph{dual} transport plan 
$\om' \in T(\n' , \m')$,
which can be defined via
\begin{equation}\label{om'}
	E_{\om'} = E_\om' \, ,
\end{equation}
or equivalently by 
$\om'( b' \odot a ) = \om( a \odot b' )$.


We proceed with the first step toward the fermionic version of all this, with the argument being completed in the next subsection. 
\begin{prop}
	\label{FermiKanaal}
For any $\om \in T^\g(\m , \n)$, the channel $E_\om$ is the unique function $\ca \to \cb$ such that in terms of \eqref{FermTPfromUs},
\begin{equation}\label{FdiagE}
	\om_{\ftps} = \d^{\F}_\n \circ (E_\om \ftp \id_{\cb^\wr}), 
\end{equation}
where 
$E_{\om} \ftp \id_{\cb^\wr} : \ca \ftp \cb \to B(\cg_\cb)$ 
is given by 
\begin{equation*}
	E_\om \ftp \id_{\cb^\wr }(a \ftp b^\wr  ) = E_\om(a) \ftp b^\wr.
\end{equation*}
In addition, $E_\om$ is even, and consequently
\[
\om_{\ftps} \in T^{\F}(\m , \n).
\]
\end{prop}

\begin{proof}
 Using \eqref{UrepToFrep} along with Subsection \ref{SubsecTrivgrad} and (\ref{Eformule}), we have 
\begin{align}
	\om_{\ftps}(a\ftp b^{\wr }) 
	&=\left\langle 
	\Om ,
	\pi _{\om}(a\odot 1_{\cb'})
	h_{\om}^{1/2}\pi _{\om }(1_{\ca}\odot \g _{\cb}^{-1/2}(b^{\wr })) \Om
	\right\rangle 
	\notag\\
	&=\left\langle 
	P_{\nu }\Omega ,
	\pi _{\om}(a\odot 1_{\cb'})
	h_{\nu }^{1/2}\pi _{\nu'}^{\om}(\g_{\cb}^{-1/2}(b^{\wr}))\Omega
	\right\rangle  
	\notag\\
	&=\left\langle 
	\Lambda _\n ,
	u_\n^* P_\n \pi _\om(a\odot 1_{\cb'})u_\n
	g_{\cb}^{1/2}\g_{\cb}^{-1/2}(b^{\wr})\Lambda _{\nu }
	\right\rangle  
	\notag\\
	&=\left\langle 
	\L_\n ,
	E_{\om}(a)b^{\wr }\Lambda _{\nu }
	\right\rangle  
	\notag\\
	&=\d _\n^{\F}\circ (E_{\om}\ftp \id_{\cb^{\wr}}) ( a\ftp b^\wr )
	\notag
\end{align}
as required. 
As $\cb^\wr \L_\n$ is dense in $\cg_\cb$, while $\L_\n$ is separating for $\cb$, the value $E_\om(a)$ is indeed uniquely determined by 
$
\left\langle 
\Lambda _{\nu} , E_{\om}(a)b^{\wr }\Lambda _{\nu }
\right\rangle
$
for $b^\wr$ ranging over the whole of $\cb^\wr$. 

That $E_\om$ is even, follows from
$
\omega(\gamma_\ca(a) \odot b')
=
\omega(a \odot \gamma_\cb'(b'))
$
in Definition \ref{oordPlan}, by rewriting it as
$$
\langle \L_\n , E_\om(\g_\ca(a)) b' \L_\n \rangle
=
\langle \L_\n , E_\om(a) \g'_\b(b') \L_\n \rangle
=
\langle \L_\n , \g_\b(E_\om(a)) b' \L_\n \rangle,
$$
which gives $E_\om \circ \g_\ca =  \g_\cb \circ E_\om$ by the same uniqueness argument above.
Lastly we show that $\om_{\ftps} \in T^{\F}(\m , \n)$. 
It is a state by Proposition \ref{propUrepToFrep}.
The required coupling properties (\ref{marginals}) follow from (\ref{FdiagE}), 
as does (\ref{oordPl}) by Definition \ref{verwDuaalDef} (and Example \ref{(anti)iso}, ensuring the existence of the twisted dual).
\end{proof}

A key role of this channel $E_\om$ is that it can be used to characterize transport plans between $\bz_2$-graded systems, in the same way as for the usual case in \cite[Theorem 4.1]{DS2}. 
We provide a proof directly in the $\bz_2$-graded case along the lines of evenness above, 
avoiding reliance on Hilbert space representations of the dynamics as contractions. 
This ensures that we can cover u.p. maps without further assumptions (like Schwarz inequalities).

\begin{prop}\label{fermBalE}
	Given two $\bz_2$-graded systems $\va$ and $\vb$ and a (usual) transport plan 
	$\om \in T^\g(\m, \n )$, we have
	$$
	\va \om_{\ftps} \vb
	$$
	if and only if (recalling Notation \ref{SupInd})
	$$
	E_\om \circ \a = \b \circ E_\om \, .
	$$
\end{prop}

\begin{proof}
	The condition $\va \om_{\ftps} \vb$ (Definition \ref{TF(A,B)}) can be expressed as
	$$
	\langle \L_\n , E_\om(\a(a)) b^\wr \L_\n \rangle
	=
	\langle \L_\n , E_\om(a) \b^\wr(b^\wr) \L_\n \rangle
	=
	\langle \L_\n , \b(E_\om(a)) b^\wr \L_\n \rangle
	$$
	for all $a\in\ca$ and $b^\wr \in \cb^\wr$, 
	using Proposition \ref{FermiKanaal} followed by Theorem \ref{twDuSys} (and Proposition \ref{DuaEkw}).
\end{proof}
Given $\bz_2$-graded systems $\va$ and $\vb$, and a transport plan 
$\om \in T^\g(\m, \n )$, this proposition tells us that 
 $\va \om_{\ftps} \vb$ is equivalent to $\va \om \vb$ (see Definition \ref{T(A,B)}), as both are characterized by 
 $E_\om \circ \a = \b \circ E_\om$ 
 (for $\va \om \vb$ as special case of Proposition \ref{fermBalE} when ignoring the grading).
 
 \begin{remark}\label{diag&IdCh}
 	A simple point worth noting is that $\d _\n^{\F}$ corresponds to the identity channel $\id_\cb$. This is ultimately related to the Wasserstein distance between a state and itself being zero. See \cite[Proposition 6.1]{D22} for the usual case.
 \end{remark}

\subsection{The one-to-one correspondence between fermionic and usual transport plans}
\label{SecFTPvsTP}
We now turn to the key result so far in the paper. 
%
%
It is the culmination of the work in Section \ref{SecCR} as well as this section, and will be of cardinal importance in the next section, in particular through Corollary \ref{sigma1-to-1}.

\begin{theorem}\label{Thm1-to-1}
	The state $\om_{\ftps}$ defined in (\ref{FermTPfromUs}),
	leads to a well defined bijection
\begin{equation}\label{1-to-1}
	f : T^\g(\m , \n) \rightarrow T^{\F}( \m ,\n )
	  : \om \mapsto \om_{\ftps} \, ,
\end{equation}
with inverse given by $f^{-1}(\om) = \om_{\odot}$
for all $\om \in T^{\F}(\m ,\n )$.
\end{theorem}
\begin{proof}
For any $\om \in T^\g(\m , \n)$, Proposition \ref{FermiKanaal} tells us that indeed $\om_{\ftps} \in T^{\F}(\m , \n)$.
Thus (\ref{1-to-1}) is well-defined.
In addition, Proposition \ref{FermiKanaal} says that $E_\om$ is uniquely determined by $\om_{\ftps}$, while $\om$ is determined by $E_\om$ through (\ref{omUitEom}), proving that $f$ is injective.

For any $\om \in T^{\F}(\m , \n)$, Proposition \ref{propFrepToUrep} tells us that $\om_\odot \in T^\g(\m , \n)$.
Hence $f(\om_\odot)$ is well-defined. 
Moreover, again using \eqref{Eformule}  for $E_{\om_{\odot }}$, in terms of the projection 
$P_\n$ of $\ch_\om$ onto $\ch^{\om_\odot}_\n$,
together with Proposition \ref{FermiKanaal}, Lemma \ref{Vertaalde_u} and \eqref{Fu},
it follows that 
\begin{eqnarray*}
	(\omega _{\odot })_{\ftps}(a\ftp b^{\wr }) 
	&=&\d_\n^{\F}(E_{\om_{\odot }}(a)\ftp b^\wr ) \\
	&=&\left\langle 
		  	\L_\n ,
			u_\n^* P_\n \pi _{\om_\odot}(a \odot 1_{\cb'})u_\n b^\wr \L_\n
			\right\rangle  \\
	&=&\left\langle 
			\L_\n ,
			v_\n^* P_\n \p_\om^{\Us} (a \odot 1_{\cb'})v_\n b^\wr \L_\n
			\right\rangle  \\
	&=&\left\langle 
			\Om , 
			\pi_\om^{\Us}(a \odot 1_{\cb'}) \pi _{\n^\wr}^\om(b^{\wr }) \Om
			\right\rangle,
\end{eqnarray*}
since $P_\n v_\n \L_\n = P_\n \Om = \Om$. 
By the definitions \eqref{fermRepTrans} and \eqref{erfnu} of $\pi_\om^{\Us}$ and $\pi _{\nu ^{\wr }}^{\omega }$ respectively, 
and because
$\pi _{\n ^\wr}^\om (b^\wr) \Om \in \ch_\om^0$,
we obtain  
\begin{eqnarray*}
	(\om_\odot)_{\ftps} ( a \ftp b^\wr ) 
	&=&
	\left\langle 
	\Om , \pi_\om^{0}(a\ftp 1_{\cb'}) \pi_\om^0 (1_\ca \ftp b^\wr )\Om
	\right\rangle  \\
	&=&
	\left\langle 
	\Om , \pi _{\omega }^{0}(a\ftp b^{\wr })\Omega 
	\right\rangle  \\
	&=&\omega (a\ftp b^{\wr }),
\end{eqnarray*}
by \eqref{GNS}, that is, 
	$(\om_{\odot })_{\ftps} = \om$.
Thus we have shown \eqref{1-to-1} to be surjective, hence bijective with inverse as claimed.
\end{proof}
\begin{rem}\label{FGNSb}
	Combining this theorem with Proposition \ref{propUrepToFrep}, we see that the cyclic representation associated to any 
	$\om \in T^{\F} ( \m , \n )$
	is indeed as bounded operators, as promised in Remark \ref{begrensdeBelofte}.
	That is, the cyclic representation $\p_\om^0$ from Subsection \ref{SubSecFTrPl}
	can be written as
	$$
	\p_\om : \ca \ftp \cb^\wr \to B(\ch_\om)
	$$
	by simply extending $\p_\om^0 (c)$ continuously to $\ch_\om$.
	Consequently $\om$ itself is bounded. In our translation approach, this is ultimately because $\om_\odot$ is.
\end{rem}

Theorem \ref{Thm1-to-1}
allows us to complete the representation of fermionic transport plans in terms of channels started in the previous subsection,
now telling us that \emph{every} fermionic transport plan $\om$ is represented by a corresponding even channel:

\begin{corollary}\label{ferE}
	For $\om \in T^{\F} (\m , \n)$ there is a uniquely determined function 
	$E_\om : \ca \to \cb$ 
	such that
		$$
		\om = \d^{\F}_\n \circ (E_\om \ftp \id_{\cb^\wr}), 
		$$
	This function is the normal even channel given by 
	$$E_\om = E_{\om_\odot},$$
	and satisfies $\n \circ E_\om = \m$.
\end{corollary}

\begin{proof}
	The existence of such a channel follows from Proposition \ref{FermiKanaal} and Theorem \ref{Thm1-to-1} by simply setting $E_\om := E_{\om_\odot}$, and its uniqueness as a function follows from the same. 
	Since 
	$\om_\odot \in T^\g (\m , \n)$, 
	we know that $E_\om$ is even 
	(Proposition \ref{FermiKanaal}), 
	while normality and $\n \circ E_\om = \m$ was recalled in Subsection \ref{SubsecEvsom}.
\end{proof}
This in turn allows the following convenient restatement of Proposition \ref{fermBalE}, to be used in Subsection \ref{OndAfdWSim}.
\begin{cor}\label{ferBalEDirek}
	Given two $\bz_2$-graded systems $\va$ and $\vb$ and a fermionic transport plan 
	$\om \in T^{\F}(\m, \n )$, we have
	$$
	\va \om \vb
	$$
	if and only if
	$$
	E_\om \circ \a = \b \circ E_\om \, .
	$$
\end{cor}
\begin{rem}
	The bijection in Theorem \ref{Thm1-to-1} can be expected to preserve structure. For example, convex combinations of transport plans, since indeed for any
	$\om, \psi \in T^\g ( \m , \n )$ and $0 \leq \la \leq 1$ 
	we see from \eqref{omUitEom} that
	$$
	E_{\la\om + (1 - \la)\psi} = \la E_\om + (1 - \la)E_\psi \, ,
	$$
	while
	$
	\om_{\ftps} = \d^{\F}_\n \circ (E_\om \ftp \id_{\cb^\wr}).
	$
	For our purposes in Section \ref{SecFW2}, however, the bijection itself is sufficient.
\end{rem}
Corollary \ref{ferE} also allows us to define transport plans corresponding to the duals of $E_\om$. For example, the KMS dual
of $E_\om$ corresponds to the transport plan 
$\om^\s \in T(\n , \m )$,
which can be defined via
\begin{equation}\label{om^s}
	E_{\om^\s} = E_\om^\s \, .
\end{equation}

The theorem implies the corresponding results for systems:

\begin{corollary}
\label{sigma1-to-1}
The maps
\begin{equation}
	\label{1}
	T (\va^\g , \vb^\g) 
	\to T^{\F} (\va , \vb)
	: \om \mapsto \om_{\ftps}, 
\end{equation}
\begin{equation*}
	T_\s (\va^\g , \vb^\g) 
	\to T^{\F}_\s (\va , \vb)
	: \om \mapsto \om_{\ftps} 
\end{equation*}
and
\begin{equation*}
	T_{\s\s} (\va^\g , \vb^\g) 
	\to T^{\F}_{\s\s} (\va , \vb)
	: \om \mapsto \om_{\ftps} 
\end{equation*}
are well defined bijections.
\end{corollary}
\begin{proof}
	We simply have to restrict the bijection $f$ from Theorem \ref{Thm1-to-1} progressively, verifying that we obtain a well-defined surjection in each case.
	
Assuming that $\om \in T( \va^\g , \vb^\g )$,
that is $\om \in T^\g(\m , \n)$
and
\[
E_\om \circ \a = \b \circ E_\om
\] 
by the non-graded case of Proposition \ref{fermBalE}, 
we conclude that 
$\om_{\ftps} \in T^{\F} (\va , \vb)$ by Theorem \ref{Thm1-to-1}  and Proposition \ref{fermBalE}, which means (\ref{1}) is well-defined. 

On the other hand, for $\om \in T^{\F}( \va , \vb )$ we similarly have 
$\om \in T^{F}(\m , \n)$ 
and 
$E_\om \circ \a = \b \circ E_\om$, 
while $E_{\om_\odot} = E_\om$, hence 
$\om_\odot \in T(\va^\g , \vb^{\g})$.
Hence (\ref{1}) is surjective, thus bijective.

The other two follow in a similar way, keeping in mind that $\g_\ca^\s = \g_\ca$ and 
$\g_\cb^\s = \g_\cb$ 
(because of \eqref{Jg=gJ} and Example \ref{(anti)iso})
when considering the KMS duality requirements in $T_{\s\s}$.
\end{proof}

Corollary \ref{ferE} was used in the proof of this key corollary,
but will also be of great direct value in expressing the cost of transport subsequently leading to Wasserstein distances between $\bz_2$-graded systems treated in the next section.

\section{Fermionic Wasserstein distances}
\label{SecFW2}


Given the setup developed in the previous sections, it is now a fairly straightforward matter to translate known results about Wasserstein distances in the case of the usual tensor product to the Fermi case. 

To define fermionic Wasserstein distances, we first need to define a suitable cost of transport.
Our focus is on the quadratic case, extending quadratic Wasserstein distances in the classical case where the probability space is a bounded closed subset of $\br^n$.
As will be seen, we can define a cost leading to fermionic Wasserstein distances in the same way as for the usual tensor product. Subsequently fermionic Wasserstein distances are defined in terms of fermionic transport plans, but using Corollary \ref{sigma1-to-1} they can be reformulated in terms of usual transport plans. This will allow us to derive the metric properties. 
Our framework for fermionic transport plans in the preceding sections has been set up carefully with this goal in mind.


Throughout this section $\va$ and $\vb$ will denote $\bz_2$-graded systems as defined in Definition \ref{stelsel}. 
It will either be clear from context, or stated explicitly, whether the $\bz_2$-grading is being taken into account or not. 
All $\bz_2$-graded systems in this section 
are assumed to have the same
number of coordinates $d$ and index set $\U$, 
both taken as fixed throughout.
It is only now that the coordinates
$k_1 ,..., k_d \in \ca$ and $l_1 ,..., l_d \in \cb$
become relevant, serving as a generalization of classical coordinate functions on a bounded closed subset of $\br^d$, to be used in the definition of a quadratic Wasserstein distance between systems. 
In practice one could choose them as corresponding observables in the two systems respectively, for example 
1-particle state number operators (i.e., occupation numbers) in multi-particle  systems.

A simple way to define the cost of a fermionic transport plan, motivated by the usual case \cite[Section 7]{DSS}, is in terms of the corresponding channel $E_\omega$ from Corollary \ref{ferE}:

\begin{definition}
	\label{I}
	Assuming that $\va$ and $\vb$ are $\bz_2$-graded, the \emph{cost} of a fermionic transport plan 
	$\om \in T^{\F} (\va,\vb)$ is defined as
	\begin{equation}
		\label{ferKoste}
		I_{\va,\vb}(\om) 
		= 
		\sum_{i=1}^{d}
		\left[  \m(k_i^* k_i) + \n(l_i^* l_i) - \n(E_\om (k_i)^*l_{i}) - \n(l_i^* E_\om (k_i))
		\right]. 
	\end{equation}
\end{definition}
\begin{rem}
	To gain intuition about this form for the cost, view the transport plan $\om$ as dynamics, or a process, taking an operator
	$a$ (self-adjoint in case of an observable) of the first system, to an operator $E_\om (a)$ of the second.
	Then one could view 
	\[
	\om (a^* \odot b) = \n ( E_\om (a)^* b )
	\]
	as a measure of correlation between $a$
	after transport, and an operator $b$ of the second system.
	We can thus think of the transport plan $\om$ itself in these terms.
	To clarify it further, specialize to
	\[
	\om (a^* \odot a) = \n (E_\om (a)^* a) \, ,
	\]
	where we now consider transport on the same algebra (classically, within the same space),
	viewing $\n (E_\om (a)^* a)$ as a measure of correlation between $E_\om (a)$ and $a$. 
	Then we expect higher correlation to correspond to lower cost.
	That is: if $E_\om$ does less, then $\om$ costs less.
	
	One can also build intuition about symmetry of a Wasserstein distance from this correlation view.
	Expressing cost in terms of correlations 
	$
	\n ( E_\om (a)^* b ) 
	$
	or 
	$
	\n ( b^* E_\om (a) ) \, ,
	$
	%
	%
	we could expect that sufficient symmetry of correlation under an appropriate reversal of $\om$ as dynamics, may lead to 
	symmetry of a Wasserstein distance. Plausibly we could require this in terms of the KMS dual,
	\begin{equation}\label{KorSim}
	\n (E_\om (a)^* b) = \m (a^* E^\s_\om (b) ) \, ,
	\end{equation}
	as the KMS dual is the representation of the (twisted) dual in terms of the algebras themselves.
	The transport plan $\om^\s$ corresponds to $E_\om^\s$, as in \eqref{om^s}, in this sense thus being the reverse transport plan to $\om$.
	The form $\m ( E^\s_\om (b) a )$
	on the right would be more precise in corresponding to $\om^\s$ in the arguments above,
	but below, \eqref{KorSim} will be seen to be sufficient.
	Imposing \eqref{KorSim},
	we indeed immediately find symmetry of the general cost term
	$$
	Q_{a,b} (\om) 
	:=
	\m (a^* a) + \n(b^* b) - \n( E_\om (a)^* b ) - \n( b^* E_\om (a) )
	=
	Q_{b,a} (\om^\s) \, .
	$$
	Since this holds for all $a$ and $b$, it is strictly speaking more than we need for symmetry of a Wasserstein distance,
	but \eqref{KorSim} is equivalent to the condition
	$
	( \ca , \s^\m , \m ) \, \om \, ( \cb , \s^\n , \n ) 
	$
	defining $T^{\F}_\s ( \va , \vb )$ (see \cite[Section 5]{D22}),
	which is natural in our system context.
	To ensure symmetry of a Wasserstein distance between systems, however, we also need to know that $\om^\s$ is allowed as a transport plan from $\vb$ to $\va$, that is, we can expect to require $\vb \om^\s \va$.
	Because of Corollary \ref{ferBalEDirek}, this is equivalent to the additional condition $\va^\s \om \vb^\s$ set in the definition of 
	$T^{\F}_{\s\s} ( \va , \vb )$.
	
\end{rem}
Notice that the formula \eqref{ferKoste} for cost does not depend in any way on the tensor product being used. It depends on the transport plan only via the channel $E_\om$ representing it, which we saw in Corollary \ref{ferE} is the same for the corresponding fermionic and usual transport plans in terms of the bijection given by Theorem \ref{Thm1-to-1}.
This is a key point in our strategy to obtain fermionic Wasserstein distances.

We use the symbol
$$X$$
to denote a set of $\bz_2$-graded systems on which we intend to define Wasserstein distances.
As in the usual case, we have to optimize the cost over a particular set of fermionic transport plans in order to obtain (the square of) the corresponding fermionic Wasserstein distance on $X$:

\begin{definition}
	\label{W2}
	The \emph{fermionic Wasserstein distance} $W$\ on $X$ is defined  by
	\[
	W^{\F}(\va,\vb) := \inf_{\om \in T^{\F}(\va,\vb)} I_{\va,\vb}(\om)^{1/2},
	\]
	the \emph{modular fermionic Wasserstein distance} $W^{\F}_\s$\ on $X$ by
	\[
	W^{\F}_\s(\va,\vb) := \inf_{\om \in T^{\F}_\s(\va,\vb)} I_{\va,\vb}(\om)^{1/2},
	\]
	and the \emph{KMS fermionic Wasserstein distance} $W^{\F}_{\s\s}$\ on $X$ by
	\[
	W^{\F}_{\s\s}(\va,\vb) := \inf_{\om\in T^{\F}_{\s\s}(\va,\vb)} I_{\va,\vb}(\om)^{1/2},
	\]
	for all $\va,\vb\in X$, in terms of Definitions \ref{TF(A,B)} and
	\ref{I}.
\end{definition}

\begin{remark}
	\label{Ialt}Because $E_\om = E_{\om_\odot}$,
	it follows as in \cite[Section 2 and 3]{D22} that in terms of the representations from Subsection \ref{SubsecTrivgrad}, 
	one has
	$$
	I_{\va,\vb}(\om)
	=
	\left\|  
	\pi^{\om_\odot}_\m (k) \Om - \pi^{\om_\odot}_\n (l) \Om
	\right\| _{\oplus\om} ^2,
	$$
	where 
	$
	\p^{\om_\odot}_\m (k) \Om 
	\equiv 
	\left( 
	\p^{\om_\odot}_\m (k_1) \Om ,..., \p^{\om_\odot}_\m (k_d) \Om
	\right)
	\in \bigoplus_{l=1}^n \ch_{\om_\odot}
	$, 
	with $\left\| \cdot \right\| _{\oplus\om}$ the norm on the latter,
	etc. 
	This tells us that the distances defined in Definition \ref{W2} are indeed real and non-negative.
	It also gives a more intuitive
	idea of why we expect to obtain the triangle
	inequality for Wasserstein distances from this cost. 
	Indeed, in the usual case it follows from this form using relative tensor products of bimodules. 
	However, unlike \eqref{ferKoste}, tensor products play an explicit role here, both in setting up the representations from $\om$ and in the relative tensor product, which is why we rather opt for \eqref{ferKoste} when formulating the fermionic cost. 
\end{remark}

From the definition we clearly have
\[
W^{\F}(\va,\vb) 
\leq 
W^{\F}_{\s}(\va,\vb)
\leq
W^{\F}_{\s\s}(\va,\vb)
\]
for all $\va,\vb \in X$.

%
%

We are now in a position to prove metric properties of the functions 
$W^{\F}$, $W^{\F}_{\s}$ and $W^{\F}_{\s\s}$,
as well as the existence of optimal transport plans, by appealing to the corresponding properties from the usual case studied in \cite[Section 7]{DSS}. The triangle inequality in particular will substantiate the term ``distance", though it should be kept in mind that as in the usual case (see the examples in \cite[Sec 5.3]{D23} and \cite[Sec 4.3]{DSS}), not all these fermionic Wasserstein distances need be symmetric.

The Wasserstein distances in the usual case will be denoted as 
$$
W, W_\s \text{ and } W_{\s\s},
$$
respectively, obtained when ignoring the
gradings in Definition \ref{W2}.

\begin{definition}
	An \emph{optimal} transport plan for 
	$W^{\F}_{\s\s} (\va,\vb)$
	is an 
	$\om \in T^{\F}_{\s\s}(\va,\vb)$ 
	such that
	$I_{\va,\vb}(\om)^{1/2} = W^{\F}_{\s\s}(\va,\vb)$. 
	If an optimal transport plan exists for 
	$W^{\F}_{\s\s}(\va,\vb)$ 
	for all $\va,\vb \in X$,
	then we say that optimal transport plans for $W^{\F}_{\s\s}$
	\emph{always exist}.
	Analogously for $W^{\F}$ and $W^{\F}_\s$.
\end{definition}

We remind the reader that for any set $Y$, 
a real-valued function $\r$ on $Y \times Y$ is called 
an \emph{asymmetric pseudometric} if it satisfies the triangle
inequality, $\r(x,y) \geq 0$ and $\r(x,x) = 0$ for all $x,y\in Y$. 
We call $\r$ a \emph{pseudometric}
if furthermore $\r(x,y) = \r(y,x)$ for all $x,y \in Y$. 

\begin{theorem}
	\label{metries}
	The functions $W^{\F}$ and $W^{\F}_\s$ are asymmetric
	pseudometrics, 
	while $W^{\F}_{\s\s}$ is a pseudometric. 
	In addition, optimal transport plans always exist for 
	$W^{\F}$, $W^{\F}_\s$ and $W^{\F}_{\s\s}$.
\end{theorem}

\begin{proof}
	By Corollary \ref{ferE} we know that for all $\om \in T^{\F}(\m , \n)$,
	$$E_\om = E_{\om_\odot}$$
	for the corresponding usual transport plan $\om_\odot \in T(\va , \vb)$.
	Therefore  \eqref{ferKoste} is exactly the cost used for the usual tensor product in \cite{D22, D23, DSS} when applying it to $\om_\odot$  instead of $\om$.
	Consequently, for a (usual) transport plan between not necessarily $\bz_2$-graded systems  $\va$ and $\vb$,
	$$\om \in T(\va , \vb),$$
    we can still use the notation 
		\begin{equation}\label{Koste}
			I_{\va,\vb}(\om) 
			= 
			\sum_{i=1}^{d}
			\left[  \m(k_i^* k_i) + \n(l_i^* l_i) - \n(E_\om (k_i)^*l_{i}) - \n(l_i^* E_\om (k_i))
			\right]
			\end{equation}
	without causing confusion. 
	In particular, given $\bz_2$-graded systems  $\va$ and $\vb$, and any
	$$\om \in T^{\F}(\va,\vb),$$ 
	we then have
	$$
	I_{\va,\vb} (\om) = I_{\va^\g,\vb^\g} (\om_\odot).
	$$
	Therefore, by Corollary \ref{sigma1-to-1}, we have
	\[
	W^{\F}(\va,\vb) 
	  = \inf_{\om \in T(\va^\g,\vb^\g)} I_{\va^\g,\vb^\g}(\om)^{1/2} 
   	  = W(\va^\g,\vb^\g),
	\]
	\[
	W^{\F}_\s (\va,\vb) 
	  = \inf_{\om \in T_\s(\va^\g,\vb^\g)} I_{\va^\g,\vb^\g}(\om)^{1/2}
	  = W_\s (\va^\g,\vb^\g),
	\]
	and 
	\[
	W^{\F}_{\s\s} (\va,\vb) 
	  = \inf_{\om\in T_{\s\s}(\va^\g,\vb^\g)} I_{\va^\g,\vb^\g}(\om)^{1/2}
	  = W_{\s\s} (\va^\g,\vb^\g) 
	\]
	for all $\va,\vb\in X$,
	which expresses the fermionic Wasserstein distances in terms of the corresponding usual Wasserstein distances from \cite[Section 7]{DSS}. 
	The required properties now follow from the corresponding properties of the usual case, \cite[Theorem 7.10]{DSS}.
\end{proof}

The remaining property, namely ``faithfulness", will be obtained next in the case of appropriate coordinates, 
also taking account of the fact that we allow systems on different algebras with different coordinate systems. 
By faithfulness of a pseudometric $\r$ as above, we of course mean that $\r(x,y) = 0$ for any $x,y \in Y$ implies $x = y$, which would tell us that $\r$ is in fact a metric.
For an asymmetric pseudometric the natural condition making it an asymmetric metric may instead be that $\r(x,y) = 0$ together with $\r(y,x) = 0$ should imply $x = y$.
Developing asymmetric metrics along these lines is discussed at length in \cite{Men13, Men14}.
 	
%

%
%
%

To prepare the ground, we introduce a refinement of the coordinates:

\begin{definition}
	A system $\va$, whether $\bz_2$-graded or not, is called \emph{hermitian} if 
	$\{k_1^* ,..., k_d^*\} = \{k_1 ,..., k_d\}$.
\end{definition}

Note that here we do not require $k_i^* = k_i$. 
To obtain faithfulness of our fermionic Wasserstein distances in the hermitian case, we again employ the translation strategy, using the corresponding result from \cite{DSS}. 
Since we allow systems on different von Neumann algebras, we have to formulate it via isomorphism of systems, as was done in \cite{DSS}. 

\begin{definition}
	\label{isoDef}
	Two $\bz_2$-graded systems $\va$ and $\vb$ are called \emph{isomorphic} if
	there is a $*$-isomorphism $\io : \ca \to \cb$ such that
	$\io \circ \g_\ca = \g_\cb \circ \io$,
	$\io \circ \a = \b \circ \io$, 
	$\n \circ \io = \m$, and
	$\io(k_i)=l_i$ for $i = 1,...,d$.
	Then $\io$ is called an isomorphism of $\va$ and $\vb$.
\end{definition}

The point is that the $*$-isomorphism $\io$ in this definition simply preserves all the structure in the definition of a $\bz_2$-graded system, hence $\va$ and $\vb$  are at most different representations of the same $\bz_2$-graded system.
If the gradings in Definition \ref{isoDef} are taken to be trivial, we recover isomorphism of systems as defined in \cite[Definition 8.12]{DSS}.
As the latter is the special case of trivial grading, and we focus on the graded setting, no confusion will arise about the type of isomorphism in question, though we'll emphasize the point in the proofs below.

First up is faithfulness of $W^{\F}$, which we need to prove in some detail, as this aspect was only covered for $W_\s$, and thus effectively also for $W_{\s\s}$, in \cite{DSS}. Similar to the latter paper, $W^{\F}_\s$ and $W^{\F}_{\s\s}$ are the more important cases  for us, but it is instructive to contrast them with $W^{\F}$.

As mentioned above, keep in mind that if there is a possibility that $\r$ may not be symmetric, then it is natural to postulate both $\r(x,y)=0$ and $\r(y,x) = 0$ as requirements to imply $x = y$. 
This is indeed what is done in the next proposition in the case of $W^{\F}$.

\begin{prop}
	\label{AsimMetrieseGetrouheid}
	Consider two hermitian $\bz_2$-graded systems $\va$ and $\vb$, 
	with $\ca$ generated by $k_1,...,k_d$ and $\cb$ by $l_1,...,l_d$. 
	If $W^{\F}(\va,\vb) = 0$ and $W^{\F}(\vb,\va) = 0$, then $\va$ and $\vb$ are isomorphic.
\end{prop}

\begin{proof}
	In keeping with our translation strategy, we first prove it for usual systems $\va$ and $\vb$ (with trivial gradings) using results from \cite{DSS}, and then obtain the $\bz_2$-graded case from it. Assume that $W(\va,\vb) = 0$ and $W(\vb,\va) = 0$. We need to show that $\va$ and $\vb$ are isomorphic as in Definition \ref{isoDef} for the case of trivial gradings.
	
	From the non-graded case of Theorem \ref{metries} 
	we know optimal transport plans exist, that is, 
	$\om \in T(\va , \vb)$ and $\psi \in T(\vb , \va)$
	such that 
	$I_{\va , \vb}(\om) = 0$ and $I_{\vb , \va}(\psi) = 0$.
	According to \cite[Lemma 8.10]{DSS} we then have that $E_\om$ and $E_\psi$ are $*$-homomorphisms respectively satisfying 
	$E_\om(k_i)  = l_i$ and $E_\psi(l_i) = k_i$.
	It follows that they are each other's inverses, making 
	$$\io = E_\om$$
	a $*$-isomorphism from $\ca$ to $\cb$. Since $\om \in T(\va , \vb)$, we know from the non-graded cases of Corollaries \ref{ferE} and \ref{ferBalEDirek} that $E_\om$ indeed has all the properties to make it an isomorphism of the systems $\va$ and $\vb$.
	
	Thus the result holds for trivial gradings. 
	To translate it to the general graded case, simply note that 
	$W^{\F}(\va,\vb) = 0$ and $W^{\F}(\vb,\va) = 0$ 
	means that 
	$W(\va^\g , \vb^\g) = 0$ and $W(\vb^\g , \va^\g) = 0$,
	as can be seen from the proof of Theorem \ref{metries}.
	Applying the result above to this, we indeed obtain the proposition.
\end{proof}

For $W^{\F}_\s$ and $W^{\F}_{\s\s}$ we can instead simply translate the corresponding result from \cite{DSS} for the usual case. Here only the assumption $\r(x,y)=0$ is needed, even for $W^{\F}_\s$. 
As can seen in the proof of \cite[Theorems 8.5 and 8.11]{DSS}, the modular condition in $T^{\F}_\s$ provides enough structure. 
In particular, it gives
$E_\om \circ \s^\m = \s^\n \circ E_\om$,
ensuring that we can reverse a transport plan from $\va$ to $\vb$ sufficiently well to obtain a transport plan from $\vb^\s$ to $\va^\s$. This is good enough to set up an isomorphism between $\va$ and $\vb$ despite the KMS duals $\vb^\s$ and $\va^\s$ not appearing in the definition of $T_{\s}(\va , \vb)$.
This structure is missing in Proposition \ref{AsimMetrieseGetrouheid}, which is why we required both
$W^{\F}(\va,\vb) = 0$ and $W^{\F}(\vb,\va) = 0$ 
in it.

\begin{prop}
	\label{MetrieseGetrouheid}
	Consider two hermitian $\bz_2$-graded systems $\va$ and $\vb$, 
	with $\ca $ generated by $k_1,...,k_d$ and $\cb$ by $l_1,...,l_d$. 
	If $W^{\F}_\s(\va,\vb) = 0$, then $\va$ and $\vb$ are isomorphic.
\end{prop}

\begin{proof}
	Note that  $W_\s (\va^\g,\vb^\g) = W^{\F}_\s(\va,\vb) = 0$ implies that 
	$\va^\g$ and $\vb^\g$ are isomorphic (as usual systems) 
	according to \cite[Corollary 8.14]{DSS}, which means precisely that $\va$ and $\vb$ are isomorphic as $\bz_2$-graded systems.
\end{proof}

As $W^{\F}_{\s} \leq W^{\F}_{\s\s}$, this proposition clearly holds for $W^{\F}_{\s\s}$ as well. This completes our basic theory for fermionic Wasserstein distances. 

\section{Deviation from fermionic detailed balance}
\label{SecAfwVanFB}

To illustrate the use of fermionic Wasserstein distances, we employ them to place bounds on the deviation of a $\bz_2$-graded system from satisfying fermionic detailed balance (FDB), in terms of any $\bz_2$-graded system that does satisfy it. This form of detailed balance was introduced in concrete form in \cite{D18} and studied further in an abstract setting in \cite{CDF}.

However, \cite{CDF} just gave a bare bones abstract formulation, rather than a complete development.
Here we develop the abstract von Neumann algebraic framework for FDB further by refining it, motivated by properties found in a concrete lattice example, and to fit it more precisely into the framework of standard quantum detailed balance (SQDB) convenient for our Wasserstein distances.

The core idea behind FDB
is that (a copy of) a $\bz_2$-graded system is equal to a twisted dual of the system, in step with a standard approach to usual detailed balance conditions. The problem is how to fully realize this idea in the fermionic case, which is what we investigate in the first and second subsection below.
In Subsection \ref{OndAfdSQDB} we show that FDB is in fact a case of SQDB with respect to a reversing operation. 
However, we argue that our approach via ``copying'' rather than using a reversing operation from the outset, 
is a conceptually clearer route to SQDB, 
preserving the simple view of detailed balance as dynamics being its own reverse.

\begin{sloppypar}
In Subsection \ref{OndAfdWSim} we turn to symmetries of the fermionic Wasserstein distances, 
which when combined with the triangle inequality,
lead to the mentioned bounds. The symmetries form an analogue of symmetries studied in the usual (i.e., non-graded) case in \cite[Section 8]{DSS}, and relate to a twisted dual and copying transformation of a $\bz_2$-graded systems relevant to FDB.
These symmetries are a key component of our approach to deviation from detailed balance, and illustrate the significance of studying symmetries of quantum Wasserstein distances.
\end{sloppypar}

As a simple picture, think of Fermionic Wasserstein distances as based on 
correlations allowed between systems involving indistinguishable fermions. For instance, two systems connected to the same reservoir of fermions, along the lines described in the lattice example below. This restriction on correlations is thus incorporated in our approach to deviation from FDB.

A very different approach to measure deviation from equilibrium in quantum systems in terms of a Wasserstein-type distance between states was developed in \cite{Ag}.

\subsection{Fermionic detailed balance}
\label{SubsecFDB}

Throughout this subsection we consider a $\bz_2$-graded system 
$\va = (\ca,\a,\mu,k)$,
but the coordinates $k$ won't play a role. 
They will only become relevant in Subsection \ref{OndAfdWSim}. 

FDB of $\va$ is formulated abstractly as a condition of the form
$\a^\vk = \a^\wr$.
We still need to describe precisely what is meant by this equation, with the left being a copy of $\va$'s dynamics expressed on the von Neumann algebra $\ca^\wr$, while the right is the twisted dual dynamics.


We briefly review a concrete lattice example to motivate our abstract setup for FDB. 
In particular, we point out properties of copying in the concrete case which were not discussed in \cite{D18}, nor built into the theory in \cite{CDF}, but will become relevant in the theory to follow.
\begin{tra} 
%
%
%
%
	Consider a finite dimensional or separable infinite dimensional
	Hilbert space $h$, as the state space (of pure states) of a single fermion in a finite or countable lattice.
	Denote the resulting Fermi 
	Fock space as $\cg$, with $\langle \cdot,\cdot\rangle$ as inner product,
	and $f_{\varnothing }$ as
	the vacuum vector. 
	
	The lattice $L$ indexes an orthonormal basis for $h$, 
	for example a discrete set of momentum states $(e_{l})_{l\in L}$ in finite volume, and we use the notation
		$a_{l}^{\dagger}$ and $a_{l}$, with $l\in L$,
	for the corresponding creation and annihilation operators. 
	Refer for instance to \cite[Section 5.2]{BR2} for a mathematical treatment of this. 
	We also write
	\begin{equation*}
		f_{(l_{1},\ldots,l_{n})}:=a_{l_{1}}^{\dagger}...a_{l_{n}}^{\dagger}f_{\varnothing}
	\end{equation*}
	for $l_{1},\ldots,l_{n}\in L$, and $n\geq 1$. For sequences $s=(s_1, \ldots, s_m)$ and $t=(t_1, \ldots, t_n)$ in $L$, we write
	$st := (s_1,\ldots,s_m , t_1,\ldots,t_n) \in D_L$. 
	In case $t=\varnothing$ or $s=\varnothing$, one has $st=s$ or $st=t$, respectively.
	
	For any subset $M$ of $L$, let $D_M$ be a set of finite sequences
	$(s_1,...,s_m)$ in $M$, for $m=0,1,2,3,...$, with $s_q\neq s_r$ when
	$q\neq r$, such that each finite subset of $M$ corresponds to exactly one
	element of $D_M$. The empty subset of $M$ corresponds to the sequence with
	$n=0$, denoted by $\varnothing\in D_M$. 
	Note that $D_M$ is countable for infinite $M$, 
	and that $(f_s)_{s\in D_L}$ is an orthonormal basis for $\cg$. 
	
	Let 
	$$\ca(M)$$
	be the von Neumann algebra generated by $\{ a_l : l\in M \}$.
	The goal is to study a system on $\ca(M)$, the latter thus playing the role of $\ca$ in the abstract theory.
	At least algebraically, 
	$\ca(M) \ftp \ca(L \backslash M)$ is then the algebra for the entire lattice $L$.
	In addition, the remaining part $L \backslash M$ of the lattice will soon provide us with the twisted commutant of $\ca(M)$, 
	placing us suitably in the abstract framework of the paper.
	To achieve this, let $\io : M \to L$ be an injection such that 
	$M \cap \io(M) = \varnothing$, 
	where we of course assume that $M$ is chosen that this is possible.
	For our purposes we assume that $L = M \cup  \io(M)$,
	thus we work in terms of
		$\iota :M \to L\backslash M$.  \label{komplBij}
	
	One realization of this consists of two identical but mutually non-interacting lattice systems (e.g., bounded cavities allowing only discrete sets of 1-particle momentum states), both connected to the same large reservoir of indistinguishable fermions. The one lattice is $M$, the other $L \backslash M$, while $\io$ maps between corresponding points in the two copies.
	
	
	Consider a faithful normal state $\m = \tr(\r_M\,\cdot)$, where in this example $\r_M$ is taken as a diagonal density matrix
	\begin{equation*}
		\r_M = \sum_{s\in D_M}p_{s}f_{s}\Join f_{s}\,,  \label{roI}
	\end{equation*}
	$(p_{s})_{s\in D_M}$ being non-zero probabilities, with  
	$x\Join y\in B(\cg)$ defined as
		$(x\Join y)z := x \left\langle y,z \right\rangle = \left\langle y,z \right\rangle x$
	for all $x,y,z \in \cg$.
	Construct the corresponding fermionic entangled vector
	\begin{equation*}
		\L_\m := \sum_{s \in D_M} p_s^{1/2} f_{s\io (s)} \in \cg  
	\end{equation*}
	where 
	$\io (s)=(\io (s_1),...,\io (s_m))$ if $s = (s_1,...,s_m)$.
	It is straightforward to show that
	$$\m(a) = \left\langle \L_\m , a\L_\m \right\rangle$$
	for all $a \in \ca(M)$.
	Let $g :\cg \to \cg$ be the self-adjoint unitary defined by
	\begin{equation*}
		g f_s
		=
		\left\{
		\begin{array}{cc}
			  f_{s} & \text{ \ if } s \text{ has even length}\,, \\
			-f_{s} & \text{ \ if } s \text{ has odd length}\,,
		\end{array}
		\right.  \label{tralGrad}
	\end{equation*}
	for all $s\in D_L$. This induces a $\bz_2$-grading $\g$ on $\ca(M)$, making $\m$ even.

	As described in \cite[Section 12]{CDF}, a result from \cite{BJL} (in terms of \cite{A1}, \cite{A2} and \cite{A1987}), can then be used to confirm the following:
	
	The vector $\L_\m$ 
     is cyclic and separating  for $\ca(M)$ in $\cg$, 
     so $\ca(M)$ is in standard form, and
	\begin{equation*}
		\ca(M)^{\wr }=\ca(L\backslash M). 
	\end{equation*}
	

	In \cite{D18}, fermionic detailed balance was defined for this lattice in the case of finite $M$. Here we briefly recall the main points, though without the assumption that $M$ is finite.
	We consider a channel
	$\a :\ca(M) \to \ca(M)$, which we aim to copy to $\ca(M)^\wr$.
	To do this, let 
	$$\vk : \ca(M) \to \ca(M)^\wr$$ 
	be the $*$-isomorphism satisfying 
	\begin{equation*}
		\vk (a_{l}) = a_{\io (l)}\,, \quad l\in M\,.   
	\end{equation*}
	One can check that it is given by 
	$\vk(a) = KaK^*$ 
	where the unitary $K : \cg \to \cg$ is defined through
	$$Kf_{s \io(t)} = (-1)^{ (|-s| + 1) |t| }  f_{t \io(s)}$$
	for all $ s, t, \in D_M$, with $| \cdot |$ indicating the length of a string.
	Then we take
	\begin{equation*}
		\a^\vk := \vk \circ \a \circ \vk^{-1} : \ca(M)^\wr \to \ca(M)^\wr 
	\end{equation*}
	as the copy of $\a$ on $\ca(M)^\wr$ given by $\vk$.
	Motivated by standard quantum detailed balance w.r.t. a reversing operation originally studied in the series of papers \cite{FU, FR10, FU2012, FR, FR15b, FU2012b}, 
	\emph{fermionic detailed balance} (FDB) of $(\ca(M) , \a , \m)$ was defined in
	\cite{D18} by a condition which is equivalent to 
	\begin{equation}\label{einddinFFB}
		\a^\vk = \a ^\wr,
	\end{equation}
	assuming the invariance
	$$\m \circ \a = \m.$$ 
	Here we used some duality theory mentioned in \cite[Theorem 12.4]{CDF} and its proof, to rewrite \cite{D18}'s condition, as was indeed also done in \cite[Section 13]{CDF}.
	
	In the usual case, without a grading, the dual $\a'$ of $\a$ represents a reverse of the latter in terms of what one can call ``elementary transitions" in $\a$, as was argued in finite dimensions in \cite[Section 5]{DOV}.
	This is a quantum analogue of replacing each transition in the transition matrix of a classical Markov chain by its reverse. 
	Furthermore, the channel $\a$ being its own reverse is exactly (standard) quantum detailed balance as discussed in \cite[Section 5]{DOV}, in the absence of parity in the system involved. 
	At least in a finite dimensional setting one can express the dual channel directly on the algebra itself, since the commutant of the standard form of $M_n$ can also be represented as $M_n$ because of $(M_n \otimes 1_n)' = 1_n \otimes M_n $. This indeed makes direct comparison between a $M_n \to M_n$ channel and its dual possible, as was done in \cite[Example 5.2]{DSS} and \cite[Section 5]{DOV}.
	
	This is in exact analogy to the simple classical condition of equality of the probability of each transition to that of the opposite transition, as set out in \cite{DOV}.
	This singles out SQDB among other forms of quantum detailed balance as the most appropriate quantum extension of classical detailed balance, as was argued in \cite{DOV}. 
	
	In the fermionic case above we correspondingly use the twisted dual instead of the dual. However, in order to express equality of the channel $\a$ to its twisted dual $\a^\wr$, we first have to copy $\a$ to $\a^\vk$ on the twisted commutant.
	This is a heuristic explanation of the definition \eqref{einddinFFB} of FDB above, as an adaptation of the usual quantum case (and therefore also of the classical case from the point of view of \cite{DOV}) to the fermionic lattice by simply replacing the usual dual by its fermionic version.
	 
	
	However, we still need to highlight some properties of $\vk$ that, along with $K$'s formula above, were not explicitly pointed out in \cite{D18}, nor incorporated into the abstract framework of \cite{CDF}.
	
	Similar to $\a^\vk$, the copy of $\m$ on $\ca(M)^\wr$ is simply
	$$\m^\vk := \m \circ \vk^{-1}.$$
	One can check 
	that 
	$$\m^\vk = \m^\wr.$$
	Since
	$\m \circ \a = \m$, we have
	$\m^\wr \circ \a^\vk = \m^\wr$.
	
	It is easily verified that $\g(a_l) = - a_l$ by applying it to $f_s$.
	Consequently it is seen that $\vk$ is even. 
	Furthermore, 
	$K\L_\m = \L_\m$ from the formulas above, hence
	$$Ka \L_\m = \vk(a) \L_\m$$
	for all $a$ in $\ca(M)$, with $K$ thus representing $\vk$ in our standard form of $\ca(M)$.
	It can also be seen to satisfy
	$$K^2 = g$$  
	by direct calculation. 
	\hfill$\square$
\end{tra}
This example gives us a concrete picture motivating the abstract case to which we now return.
In addition, we have highlighted the properties of copying which, as we'll see in the next two subsections, are relevant to build a sensible theory for deviation from FDB in terms of fermionic Wasserstein distances between $\bz_2$-graded systems. 

Abstractly we also need to ``copy" $\a$ directly to $\ca^\wr$ by an appropriate map, enabling us to compare $\a$ to its twisted dual.
The map $\vk_\ca$ doing the copying needs to satisfy certain properties to ensure systems remain systems after copying, 
but also  properties of copying highlighted in the concrete case of a lattice above. Because of the latter, the eventual copying map we'll use, does depend to some extent on the system $\va$ being copied, in particular on the state $\m$ of the system.

To define such a copying map of a $\bz_2$-graded von Neumann algebra $\ca$ abstractly, ignoring dynamics for the moment, we therefore consider an even $*$-isomorphism
$$\vk_\ca : \ca \to \ca^\wr \, ,$$
and for any $\m \in \cf_+(\ca)$ we define a faithful normal state $\m^\vk$ on $\ca^\wr$ by
$$
\m^\vk = \m \circ \vk_\ca^{-1}.
$$
Since $\m^\vk \circ \vk_\ca = \m$,
we have a well defined unitary representation 
$$K_\ca : \cg_\ca \to \cg_\ca$$
determined by
$$K_\ca a\L_\m = \vk_\ca(a)\L_{\m^\vk},$$
with $\L_{\m^\vk}$ given by the standard form as usual,
leading to $K_\ca \L_\m = \L_{\m^\vk}$ and
$$\vk_\ca(a) = K_\ca a K_\ca^*$$
for all $a$ in $\ca$.
Since $\vk_\ca$ is even, we have 
$$K_\ca g_\ca = g_\ca K_\ca 
\quad \text{and} \quad 
K_\ca g_\ca^{1/2} = g_\ca^{1/2} K_\ca.$$
Consequently, since $\ca^\wr = \vk_\ca(\ca) = K_\ca \ca K_\ca^* \,$,
$$
\ca = \ca^{\wr\wr} 
      = \g_{\ca^\wr}^{1/2}( (K_\ca \ca K_\ca^*)' ) 
      = K_\ca \g_\ca^{1/2}(\ca') K_\ca^*
      = K_\ca \ca^\wr K_\ca^*,
$$
which means that
$$
\vk_{\ca^\wr} : \ca^\wr \to \ca :
a^\wr
\mapsto
K_\ca a^\wr K_\ca^*
$$
is a well-defined even $*$-isomorphism. 
In addition, we then define the normal faithful state $\m^{\vk\vk}$ on $\ca$ by
$$
\m^{\vk\vk} = \m^\vk \circ \vk_{\ca^\wr}^{-1}.
$$
This puts us in a position to introduce the basic definition which was not yet present in \cite{CDF},
needed for our complete definition of FDB.
\begin{definition}
	\label{KopDef}
	For a $\bz_2$-graded von Neumann algebra $\ca$, 
	an even $*$-isomorphism 
	$\vk_\ca : \ca \to \ca^\wr$ 
	will be called a \emph{copying map} for $\ca$,
	with 
	$\m^\vk \in \cf_+(\ca^\wr)$ then the corresponding \emph{copy} of 
	$\m \in \cf_+(\ca)$. 
	If in addition 
	$$\m^\vk = \m^\wr$$ 
	and
	$$\vk_{\ca^\wr}  \circ \vk_\ca = \g_\ca,$$
	we say that $\vk_\ca$ is a $\m$-\emph{copying map} for $\ca$.
\end{definition}

We split the definition above into two parts to clarify later on which properties are being used at any given point, but $\m$-copying maps are the natural structure needed to give a complete definition of fermionic detailed balance below and to execute our previously described plans.

As seen above, if $\vk_\ca$ is a copying map, then so is 
$\vk_{\ca^\wr}$.
When $\vk_\ca$ is a $\m$-copying map, the condition 
$\vk_{\ca^\wr}  \circ \vk_\ca = \g_\ca$ 
is equivalent to requiring 
$$
K_\ca ^2 = g_\ca
$$
in line with the lattice example, which implies that
$$\vk_\ca  \circ \vk_{\ca^\wr} = \g_{\ca^\wr}.$$ 
In addition $\vk_{\ca^\wr}  \circ \vk_\ca = \g_\ca$ tells us that
$$\m^{\vk\vk} = \m,$$ 
which can also be written as 
$ \m^{\vk\vk} = (\m^{\vk})^\wr$, because of $\m^\vk = \m^\wr$.
Hence $\vk_{\ca^\wr}$ 
is in turn a $\m^\wr$-copying map from 
$\ca^\wr$ 
to 
$\ca^{\wr\wr} = \ca$.

\begin{remark}
	In the abstract setup, $\m$-copying maps allows for more than just the direct copying seen in the lattice example. 
	As briefly described in Subsection \ref{OndAfdSQDB} below,
	this can be viewed as catering for parity conditions in detailed balance, but we'll stay with the ``copy" terminology.
\end{remark}

We can now complete the main goal of this subsection, namely to formulate the definition of fermionic detailed balance, refining the definition proposed in \cite{CDF} where $\vk_\ca : \ca \to \ca^\wr$ was merely assumed to be a $*$-isomorphism. 
Below we also motivate this definition conceptually in terms of our setup, based on the same core idea from both the classical and usual quantum cases, 
with recourse to arguments from \cite{DOV},
independently from (but in line with) the lattice example above.

\begin{RD}
As mentioned in the lattice example, it was argued in \cite{DOV} that the dual $E'$ of a channel in Definition \ref{duaalDef}, can be physically interpreted as the reverse of the channel, but expressed in terms of the commutants of the algebras, i.e., 
$E' : \cb' \to \ca'$.
Although the arguments in \cite{DOV} were in a finite dimensional context, we heuristically take the same reverse point of view of $E'$ in our general setup.

The expression of the reverse of $E$ as $E'$ in terms of the commutants of the algebras appears to be its ``native'' mathematical form, 
considering the natural duality relation in Definition \ref{duaalDef}.  The arguments for its existence are comparatively elementary, not involving Tomita-Takesaki theory, but making essential use of the commutants. 

We can represent the reverse of $E$ on other pairs of algebras by using appropriate copying maps.
To express it on the twisted commutants, we simply carry $E'$ over (or copy it) to the twisted commutants in a natural way using the Klein isomorphisms. That is, the reverse of $E$ expressed in terms of the twisted commutants is precisely the twisted dual
$E^\wr = \g_\ca^{1/2} \circ E' \circ \g_\cb^{-1/2}$
 of $E$ as given in Proposition \ref{DuaEkw}.
As $E^\wr$ is copied from $E'$ by the very maps defining $\ca^\wr$ and $\cb^\wr$, it fits naturally into the graded structure being used.
This representation of the reverse can in turn be copied to the original algebras $\ca$ and $\cb$ 
(see Subsection \ref{OndAfdSQDB}).
%

Returning to detailed balance, the reverse of the dynamics of $\va$ expressed on $\ca^\wr$, is given by
$\a^\wr$,
whereas the dynamics $\a$ itself expressed on $\ca^\wr$ is its copy 
$\a^\vk : \ca^\wr \to \ca^\wr$
on the latter, i.e.,
$$\a^\vk := \vk_\ca \circ \a \circ \vk_\ca^{-1}.$$
As in the non-graded quantum case (see \cite{DOV}), and indeed in a classical Markov chain as well, FDB simply states that the dynamics must be equal to its (or in our general setup, ``their") reverse.
In other words, we must have the following, 
as a refinement of \cite[Definition 13.2]{CDF}. 
\end{RD}
\begin{definition}\label{fdb}
	A $\bz_2$-graded system $\va$ with a $\m$-copying map $\vk_\ca$ for $\ca$,
	is said to satisfy \emph{fermionic detailed balance} (with respect to $\vk_\ca$), 
	or \emph{FDB}, if
	\begin{equation*}
		\a^\vk = \a^\wr.
	\end{equation*}
\end{definition}
Strictly speaking, FDB is defined with respect to a chosen $\vk_\ca$,
which can include parity.
However, 
we won't need more complete terminology  like $\vk_\ca$-FDB,
as no ambiguity will arise.

The condition 
$\vk_{\ca^\wr}  \circ \vk_\ca = \g_\ca$ 
was set above, since on the one hand it is satisfied in the lattice example, while on the other hand it plays a role in the further theory. With $\m$-copying maps we have essentially abstracted the minimum conditions from the lattice example needed to clarify the connection to standard quantum detailed balance below, and needed in the optimal transport development presented in the Subsection \ref{OndAfdWSim}.

\subsection{Standard quantum detailed balance} 
\label{OndAfdSQDB}
Although FDB was inspired by $\th$-SQDB, that is, SQDB with respect to a reversing operation $\th$, the connection wasn't completely clarified in \cite{D18} and \cite{CDF}.
On the other hand, to connect Definition \ref{fdb} to our Wasserstein distances, we need to write it in terms of systems. But defining a copy $\va^\vk$ and dual $\va^\wr$ of a system $\va$ (as will be done in the next subsection) and then simply writing 
$\va^\vk = \va^\wr$, will place undue restrictions on the coordinates $k$ of $\va$. 
Addressing this point unambiguously, involves expressing FDB as $\th$-SQDB for an appropriate reversing operation $\th$.
	The view promoted here is that Definition \ref{fdb} brings the idea of dynamics being its own inverse to the fore, while $\th$-SQDB is a less direct derived formulation.

To achieve this, a $\m$-copying map
$\vk_\ca : \ca \to \ca^\wr$
needs to be converted into a reversing operation 
$\th_\ca : \ca \to \ca$, 
which will be built into the system $\va$ as part of its dynamics.
This will allow us to incorporate the $\m$-copying map easily in the Wasserstein distances in a way suitable for our goals.
A \emph{reversing operation} on $\ca$ is a $*$-anti-automorphism 
$\th : \ca \to \ca$
such that 
$\th \circ \th = \id_\ca$.

We simply 
define 
$$
\th_\ca : \ca \to \ca
$$
corresponding to $\vk_\ca$ as
\begin{equation}\label{RevOpDef}
\th_\ca := j_\ca \circ \g_\ca^{-1/2} \circ \vk_\ca.
\end{equation}
The next proposition tells us that this is indeed a reversing operation.
\begin{prop}
	\label{omkOp}
	Given a $\m$-copying map $\vk_\ca$ for a $\bz_2$-graded von Neumann algebra $\ca$, the map $\th_\ca$ above is an even reversing operation such that
	$\m\circ\th_\ca = \m$.
\end{prop}
\begin{proof}
	It is clearly a $*$-anti-automorphism, since $j_\ca$ is. 
	By the basic definitions,
	$\m\circ\th_\ca = \m^\wr \circ \vk_\ca = \m^\vk \circ \vk_\ca = \m$.
	From $\m^\wr \circ \vk_\ca = \m$ we know (see the Appendix) that 
	$
	K_\ca J_\ca = J_\ca K_\ca.
	$
	Combining this with the commutation of both $J_\ca$ and $K_\ca$ with $g_\ca^{1/2}$, along with $K_\ca ^2 = g_\ca$,
	it follows that 
	$\th_\ca \circ \th_\ca = \id_\ca$
	as required.
    Similarly, $\th_\ca \circ \g_\ca = \g_\ca \circ \th_\ca$, that is, $\th_\ca$ is even.
\end{proof}

\begin{conv}\label{revConv}
	For transport plans to have the relevant properties with respect to $\th_\ca$ in \eqref{RevOpDef}, it will be assumed to be part of $\va$'s dynamics whenever a $\m$-copying map is present, i.e., 
for some $\u \in \U$,
$$
\a_\u = \th_\ca \, .
$$
\end{conv}
Then, in line with \cite[Definition 8.3]{DSS}:
\begin{definition}\label{revDef}
	A $\bz_2$-graded system $\va$ with a $\m$-copying map is called \emph{reversible}.
\end{definition}

\begin{parity}
	Although for us the primary role of a $\m$-copying map $\vk_\ca$ will be to copy dynamics, it may also include parity; see \cite[Section 6]{DOV}.
	That is, the $\m$-copying map allows for parity in the definition of FDB, as was mentioned in the previous subsection.
	
	In our abstract setup we cannot separate copying and parity in $\vk_\ca$, but in concrete cases this should typically be possible.
	For instance, in the lattice example in Subsection \ref{SubsecFDB}, one can see on physical grounds that $\vk$ just copies, with no parity included.
	In finite dimensions and without a grading, \cite[Section 6]{DOV} shows exactly how to separate copying and parity in $\vk_\ca$ in general.
	Abstractly, however, we'll continue to simply speak of copying, even though parity may be included.
	
It could nevertheless be of interest to investigate
whether there are further abstract conditions which can be implemented to limit the variety of $\m$-copying maps, while allowing various concrete cases of interest.
\end{parity}

With reversing operations in hand, we proceed to connect it to our Definition \ref{fdb} of FDB. The latter is stated for reversible systems, in the terminology above.
As discussed in the  previous subsection, the twisted dual $E^\wr$ in Definition \ref{verwDuaalDef} of the even channel 
$E : \ca \to \cb$ with $\n \circ E = \m$,
can be viewed as the reverse of $E$. 
The equality between dynamics $\a$ and its reverse $\a^\wr$, expresses FDB. This equality is formulated on the twisted commutant $\ca^\wr$ by copying the dynamics to the latter. We'll now see that when this is expressed on $\ca$, 
it is exactly $\th_\ca$-SQDB.

In terms of the more general situation in Definition \ref{verwDuaalDef}, and given copying maps $\vk_\ca$ and $\vk_\cb$, we introduce 
$E^\lp : \cb \to \ca$
by setting
\begin{equation}\label{Eomk}
	E^\lp = \vk_\ca^{-1} \circ E^\wr \circ \vk_\cb
\end{equation}
to copy $E^\wr$ to the original algebras. 
As $\vk_\ca$ and $\vk_\cb$ are $*$-isomorphisms, it is justified to view $E^\lp$ as a copy of $E^\wr$,
which is in turn a copy of $E'$ as explained in Subsection \ref{SubsecFDB}. 
Hence
$E^\lp$
is the reverse $E'$ of $E$ (relative to $\m$ and $\n$), but represented in terms of $\ca$ and $\cb$.

As in the lattice example, $\vk_\ca$ and $\vk_\cb$ are determined by the physical setup.
Therefore, representing the reverse of $E$ as $E^\lp$ in terms of $\ca$ and $\cb$, is expected to be mathematically less natural than $E^\wr$, rather describing a specific physical situation, including parity.

However, in order to be viewed as the reverse of $E$, a second reversal ought to lead back to $E$, that is,
\begin{equation}\label{dubbelOmk}
E^{\lp\lp} = E \, ,
\end{equation}
where of course this reverse of $E^\lp$ is taken relative to the copied states 
$\n^\wr \circ \vk_\cb$ and $\m^\wr \circ \vk_\ca$
on $\cb$ and $\ca$ respectively.
One can check that \eqref{dubbelOmk} is ensured if we assume that
$\vk_\ca$ and $\vk_\cb$
are $\m$- and $\n$-copying maps respectively,
again emphasizing the essential role of the latter property when representing relevant structures on the original algebras.

Continuing with these $\m$- and $\n$-copying maps, it is easily confirmed from the definitions that
\begin{equation*}
	E^\lp = \th_\ca \circ E^\s \circ \th_\cb
\end{equation*}
in terms of the KMS dual from \eqref{KMS}.
In particular,
$$
\th_\ca^\lp = \th_\ca \, ,
$$
as in the non-graded case (\cite[Section 8.1]{DSS}).

In terms of \eqref{Eomk}, FDB of a reversible $\va$ can therefore equivalently be expressed as 
$$
\a^\lp = \a \, ,
$$
which is in fact the definition of $\th_\ca$-SQDB. 
See \cite{FR} and \cite{BQ2} (emphasizing the von Neumann algebra $B(H)$) for this formulation of $\th_\ca$-SQDB in terms of the  
$\th_\ca$-\emph{KMS adjoint} of $\a$, defined in \cite{BQ2} as
$\th_\ca \circ \a^\s \circ \th_\ca$.
Our copying approach recovered this adjoint from a different point of view as the reverse $\a^\lp$.
To summarize the key point in terms of \eqref{RevOpDef}:
\begin{prop}
	FDB in Definition \ref{fdb} is equivalent to $\th_\ca$-SQDB.
\end{prop}
\begin{remark}
	The proof of this proposition was technically quite simple.
	However, viewing $E^\wr$, and thus its copy $E^\lp$, as (representations of) the reverse of $E$, when working with $\m$- and $\n$-copying maps, makes 
	$E^\lp = \vk_\ca^{-1} \circ E^\wr \circ \vk_\cb$ 
	more intuitive than the formula
	$\th_\ca \circ E^\s \circ \th_\cb$.
	
	In the latter $\th_\ca$ and $\th_\cb$ are not $*$-isomorphisms.
	Similarly for $j_\ca$ and $j_\cb$ in the definition of $E^\s$ from $E'$.
	They are $*$-anti-isomorphisms, hence they don't preserve all structure, and therefore do not copy like $\vk_\ca$ and $\vk_\cb$.
	Thus the direct interpretation of $E^\lp$ as a copy of $E'$
	(the native form of the reverse of $E$), is somewhat hidden in $\th_\ca \circ E^\s \circ \th_\cb$.
	Likewise, the simple role of the reverse of dynamics is not explicit in the original view of $\th_\ca$-SQDB. Our framing of quantum detailed balance in terms of copying, clarifies it conceptually by transparently bringing the simple notion of dynamics being its own reverse to centre stage.
	
	Nevertheless, the proposition above shows that this leads back to SQDB. These arguments to clarify the conceptual meaning of SQDB are in line with \cite{DOV}'s approach (though only in finite dimensions there), and are just as relevant to the usual case of trivial grading.
	In the graded case, however, we first needed the full formulation Definition \ref{KopDef} of copying.
\end{remark}

Lastly, in our setting of Wasserstein distance between systems, it is more appropriate to express FDB of a reversible system $\va$ as 
\begin{equation}\label{fdbStelsel}
	\va^\lp = \va
\end{equation}
in terms of the \emph{reverse} of $\va$ defined by
\begin{equation}\label{stelselOmk}
	\va^\lp = ( \ca , \a^\lp , \m , k )
\end{equation}
as in \cite[Definition 8.3]{DSS}, but now for the graded case. 
Although only $\a$ is relevant in this definition, in the next subsection we explore deviation from FDB in terms of distances between systems, with the other facets of a systems then also being relevant.


\subsection{Symmetries of fermionic Wasserstein distances}
\label{OndAfdWSim}

We present symmetries of fermionic Wasserstein distances, allowing us to formulate bounds on deviation from FDB in a similar way to the usual case in \cite[Sec 8.1]{DSS}. This is possible due to the refined setup for FDB developed above. The symmetries relate to the two sides of the equation 
$\a^\wr = \a^\vk$, 
as well as to the reversal appearing in the system formulation \eqref{fdbStelsel} of FDB.
The latter symmetry in particular is relevant to express deviation from FDB in terms of Wasserstein distance between systems.

We begin with twisted duals, which simply build on the twisted dual maps given by Theorem \ref{twDuSys}:
\begin{definition}
	\label{VerwDuaalStDef}
	The \emph{twisted dual} of the $\bz_2$-graded system $\va$, is the $\bz_2$-graded system
	$$
	\va^{\wr} = (\ca^\wr, \a^\wr, \mu^\wr, k^\wr)
	$$ 
	on the twisted commutant $\ca^{\wr}$ of $\ca$, 
	where we write
	$$
	k^\wr = (k_1^\wr,...,k_d^\wr)
	 \quad  \text{with}  \quad 
	k_i^\wr := \g_\ca^{1/2} \circ j_\ca(k_i^*).
	$$
\end{definition}

Consider $\bz_2$-graded systems $\va$ and $\vb$, as well as any transport plan $\om \in T^{\F}(\m , \n)$. Define a linear functional $\om^\wr$ on 
$\cb^\wr \ftp \ca$ by
\begin{equation*} 
	\om^\wr 
	=
	\d^{\F}_{\m^\wr} \circ ( E_\om^\wr \ftp \id_\ca ).
\end{equation*}
Using Definition \ref{verwDuaalDef} this simply  means that 
$$
\om^\wr(b^\wr \ftp a)
=
\om( ( a^* \ftp b^{\wr*} )^* )
=
\om(a_+ \ftp b^\wr_+ + a_+ \ftp b^\wr_- + a_- \ftp b^\wr_+ - a_- \ftp b^\wr_-)
$$
for all $a$ in $\ca$ and $b^\wr$ in $\cb^\wr$.
Note that by Theorem \ref{Thm1-to-1} 
(also keep in mind \eqref{omUitEom}, Proposition 
\ref{FermiKanaal} and Corollary \ref{ferE})
\begin{equation}\label{verwOP}
	\om^\wr = \psi_{\ftps} \in T^{\F} ( \n^\wr , \m^\wr )
\end{equation}
in terms of
$$
\psi 
:= 
\d_{\m^\wr} \circ ( E_\om^\wr  \odot \id_{ (\ca^\wr)' } )
\in
T^\g ( \n^\wr , \m^\wr ).
$$
Because of Corollary \ref{ferE} we have
\begin{equation}\label{verwE}
	E_{\om^\wr} = E_\om^\wr,
\end{equation}
hence analogously
$E_{\om^{\wr\wr}} = E_{\om^\wr}^\wr = E_\om,$ thus
\begin{equation}\label{dubVerwOP}
	\om^{\wr\wr} = \om.
\end{equation}

In terms of the transport plan \eqref{verwOP} we can state and easily prove the following simple but useful fact.
\begin{prop}
	\label{verwBalKar}
	Consider $\bz_2$-graded systems $\va$ and $\vb$, as well as any transport plan $\om \in T^{\F}(\m , \n)$. Then
	$$\va \om \vb \iff \vb^\wr \om^\wr \va^\wr.$$ 
\end{prop}
\begin{proof}
	By Corollary \ref{ferBalEDirek} and Proposition \ref{compDu} we have
	$E_\om \circ \a = \b \circ E_\om$
	if and only if 
	$\a^\wr \circ E_\om^\wr = E_\om^\wr \circ \b^\wr$.
\end{proof}
Consequently we obtain the following, to be used in deriving our first set of symmetries of fermionic Wasserstein distances.
\begin{prop}
	\label{verwBij}
	For $\bz_2$-graded systems $\va$ and $\vb$ the map 
	$$\gt_\wr : T^{\F} ( \m , \n ) \to T^{\F} ( \n^\wr , \m^\wr ) : \om \mapsto \om^\wr$$
	and its restrictions
	$$
	\gt_\wr|_{T^{\F}_\s ( \va , \vb )}  : 
	T^{\F}_\s ( \va , \vb ) \to T^{\F}_\s ( \vb^\wr , \va^\wr )
	$$
	and
	$$
	\gt_\wr|_{T^{\F}_{\s\s} ( \va , \vb )}  : 
	T^{\F}_{\s\s} ( \va , \vb ) \to T^{\F}_{\s\s} ( \vb^\wr , \va^\wr )
	$$
	are well-defined bijections.
\end{prop}
\begin{proof}
	By \eqref{verwOP} we know that $\gt_\wr$ is well-defined, as is the correspondingly defined map
	$T^{\F} ( \n^\wr , \m^\wr ) \to T^{\F} ( \m , \n ) : \psi \mapsto \psi^\wr$,
	which means that $\gt_\wr$ is bijective because of \eqref{dubVerwOP}.
	Similarly for it's restriction
	$T^{\F} ( \va , \vb ) \to T^{\F} ( \vb^\wr , \va^\wr )$ 
	because of Proposition \ref{verwBalKar}.
	
	For 
	$\om \in T^{\F}_\s (\va , \vb)$ 
	we have
	$(\ca , \s^\m , \m) \om (\cb , \s^\n , \n)$
	in addition to the conditions involved above,
	hence
	$
	(\cb^\wr , \s^{\n^\wr} , \n^\wr) 
	\om^\wr 
	(\ca^\wr , \s^{\m^\wr} , \m^\wr)
	$
	because of Proposition \ref{verwBalKar}, and \eqref{verwMod} in the Appendix.
	This means that $\om^\wr \in T^{\F}_\s (\vb^\wr , \va^\wr)$. 
	An analogous argument for $\psi \in T^{\F}_\s (\vb^\wr , \va^\wr)$ completes the proof for $\gt_\wr|_{T^{\F}_\s ( \va , \vb )}$.
	
	Lastly, if we consider 
	$\om \in T^{\F}_{\s\s} (\va , \vb)$,
	then it means we are adding the further condition 
	$\va^\s \om \vb^\s$.
	By Proposition \ref{verwBalKar} and \eqref{KMSverwDu} it follows that 
	$(\vb^\wr)^\s \om^\wr (\va^\wr)^\s$, 
	which together with the earlier assumed properties above means that
	$\om^\wr \in T^{\F}_{\s\s} (\vb^\wr , \va^\wr)$.
	Analogously starting with
	$\psi \in T^{\F}_{\s\s} (\vb^\wr , \va^\wr)$.
\end{proof}

Hence we find the following symmetries of $W^{\F}_\s$ and $W^{\F}_{\s\s}$.
\begin{theorem} \label{verwSim}
	For $\bz_2$-graded systems $\va$ and $\vb$ we have
	$$
	W^{\F}_\s (\vb^\wr , \va^\wr) =  W^{\F}_\s (\va , \vb)
	\quad \text{and} \quad 
	W^{\F}_{\s\s} (\va^\wr , \vb^\wr) =  W^{\F}_{\s\s} (\va , \vb).
	$$
\end{theorem}
\begin{proof}
	For any $\om \in T^{\F} (\va , \vb)$ we have
	$$
	E_{\om^\wr} 
	= E_\om^\wr 
	= \g_\cb^{1/2} \circ E_{(\om_\odot)'} \circ \g_\cb^{-1/2},
	$$
	by \eqref{verwE} and using Corollary \ref{ferE} and \eqref{om'}. 
	Consequently, in terms of Definition \ref{I}, and the cost \eqref{Koste} of transport plans for the usual tensor product, as well as the non-graded case
	$k_i' = j_\ca (k_i^*)$ and $l_i' = j_\cb (l_i^*)$
	of Definition \ref{VerwDuaalStDef},
	\begin{align*}
		I_{\vb^\wr,\va^\wr} & (\om^\wr)
		  =
		\sum_{i=1}^d
		\left[  
		\n^\wr (l_i^{\wr*}l_i^\wr) + \m^\wr(k_{i}^\wr{}^*k_{i}^\wr)
		-\m^\wr(E_{\om^\wr}(l_i^\wr)^*k_i^\wr) 
		- \m^\wr(k_i^\wr{}^*E_{\om^\wr}(l_i^\wr))
		\right] \\
		&  =
		\sum_{i=1}^d
		\left[  
		\n'(l_i'^*l_i') + \m'(k_i'^*k_i') - \m'(E_{(\om_\odot)'}(l_i')^*k_i') - \m'(k_i'^*E_{(\om_\odot)'}(l_i'))
		\right] \\
		&  =
		I_{\vb' , \va'}((\om_\odot)').
	\end{align*}
	If $\om \in T^{\F}_\s (\va , \vb)$, Corollary \ref{sigma1-to-1} says that
	$\om_\odot \in T_\s (\va^\g , \vb^\g) \subset T_\s (\va , \vb)$.
	Hence as in the proof of \cite[Theorem 8.5]{DSS} we have 
	$I_{\vb' , \va'}((\om_\odot)') = I_{\va , \vb}(\om_\odot)$,
	making use of the defining property of $T^{\F}_\s ( \va , \vb )$,
	not just those of $T^{\F} ( \va , \vb )$.
	Again, since $E_{\om_\odot} = E_\om$, 
	and because of \eqref{ferKoste} in Definition \ref{I}'s independence of the tensor product being used,
	we obtain
	$$I_{\vb^\wr,\va^\wr}(\om^\wr) = I_{\va , \vb}(\om).$$
	The theorem then follows from the definitions of the fermionic Wasserstein distances in Definition \ref{W2} combined with Proposition \ref{verwBij} and the usual pseudometric symmetry of $W^{\F}_{\s\s}$ given by Theorem \ref{metries}.
\end{proof}

We now proceed with corresponding symmetries in terms of copies of $\bz_2$-graded systems. 
The relevant definition being the following.
\begin{definition}
	The \emph{copy} of a $\bz_2$-graded system $\va$ on $\ca^\wr$ under a copying map $\vk_\ca$,
	is the  $\bz_2$-graded system given by
	$$\va^\vk = ( \ca^\wr , \a^\vk , \m^\vk , k^\vk ),$$
	where
	$$
	\a^\vk = \vk_{\ca } \circ \a \circ \vk_{\ca}^{-1}
	$$
	and 
	$$
	k^\vk = (k_1^\vk , ... , k_d^\vk) 
	\quad \text{with} \quad 
	k_i^\vk := \vk_\ca(k_i).
	$$
\end{definition}
If $\va$ is reversible (i.e., $\vk_\ca$ is a $\m$-copying map), then so is
$\va^\vk$ with $\m^\vk$-copying map $\vk_{\ca^\wr}$, 
from the discussion in Subsection \ref{SubsecFDB}.

\begin{remark}
	For reversible $\va$ the condition
	$\va^{\vk} = \va^{\wr}$ 
	is equivalent to 
	$\a^{\vk} = \a^{\wr}$ and $\vk_\ca (k_i) = \g_\ca^{1/2} \circ j_\ca (k_i)$,
	the former being FDB, the latter being an unnecessary restriction on the coordinates. 
	Indeed, one can check simple cases in the lattice example of Subsection \ref{SubsecFDB} to verify that
	$\vk \neq \g^{1/2} \circ j_{\ca(M)} $. 
	Hence our preference for \eqref{fdbStelsel} as a system formulation of FDB. 
	
	Note that
	$
	\va^{\wr\wr} = (\ca, \a, \mu, \g_\ca (k))
	$
	with $\g_\ca(k)$ defined coordinate-wise.
	Similarly, if $\va$ is reversible,
	$
	\va^{\vk\vk} = (\ca, \a, \mu, \g_\ca (k)).
	$
	For even coordinates we would thus have 
	$\va^{\wr\wr}  = \va = \va^{\vk\vk}$,
	but this is not needed for our goals.
\end{remark}


Now consider any $\bz_2$-graded von Neumann algebras $\ca$ and $\cb$ with copying maps  $\vk_\ca$ and $\vk_\cb$ respectively.
Define the corresponding copy of an even u.p. map
$
E : \ca \to \cb
$
by
\begin{equation} \label{Ekap}
E^\vk = \vk_{\cb} \circ E \circ \vk_{\ca}^{-1}.
\end{equation}
We are of course particularly interested in the situation 
$\n \circ E = \m,$
where $\m \in \cf_+(\ca)$ and $\n \in \cf_+(\cb)$,
in which case
$\n^\vk \circ E^\vk = \m^\vk.$
\begin{prop} \label{OPkop} 
Given these assumptions, for any transport plan $\om \in T^{\F}(\m , \n)$ we define
	$\om^\vk \in T^{\F} ( \m^\vk , \n^\vk)$
as the state
\begin{equation*} 
	\om^\vk 
	=
	\om \circ ( \vk_\ca^{-1} \ftp \vk_{\cb^\wr}^{-1} )
\end{equation*}
on 
$\ca^\wr \ftp \cb = \ca^\wr \ftp \cb^{\wr\wr}$ 
(the latter to emphasize that we are still working in the 
$\, \cdot \, \ftp \, \cdot \, ^\wr$ form)
which is also given by the formula
\begin{equation} \label{omKopForm}
	\om^\vk
	=
	\d^{\F}_{\n^\vk} \circ ( E_\om^\vk  \ftp \id_\cb ) \, ,
\end{equation}
which means that 
\begin{equation}\label{omkap} 
	E_\om^\vk = E_{\om^\vk} \, . 
\end{equation}
\end{prop}
\begin{proof}
	Clearly $\om^\vk$ is a state, as $\vk_\ca^{-1} \ftp \vk_{\cb^\wr}^{-1}$ is a $*$-isomorphism.
	Explicit calculation gives 
	$\om^\vk ( a^\wr \ftp 1_\cb ) = \m^\vk (a^\wr)$
	and
	\begin{equation*} 
		\om^\vk ( 1_{\ca^\wr} \ftp b ) 
		=
		\n^\wr ( \vk_{\cb^\wr}^{-1} (b) )
		=
		\left\langle
		\L_{\n^\vk} , b \L_{\n^\vk}
		\right\rangle
		=
		(\n^\vk)^\wr (b)
	\end{equation*}
	since 
	$\vk_{B^\wr}^{-1} (b) = K_\cb^* b K_\cb$
	and 
	$ \L_\n = K_\cb^* \L_{\n^\vk} $, 
	and by the definitions of $\n^\wr$ and $(\n^\vk)^\wr$ as in \eqref{verwToest}.
Similarly, for all $a \in \ca^\wr$ and $b \in \cb$,
\begin{equation*} 
	\d^{\F}_{\n^\vk} \circ ( E_\om^\vk  \ftp \id_\cb^{\wr\wr} ) (a^\wr \ftp b )
	=
	\left\langle
	\L_{\n^\vk} , E_\om^\vk (a^\wr) b \L_{\n^\vk}
	\right\rangle
	=
	\om ( \vk_\ca^{-1} (a^\wr) \ftp \vk_{\cb^\wr}^{-1} (b) ) ,
\end{equation*}
proving \eqref{omKopForm}. 
Thus $E_{\om^\vk} = E_\om^\vk$ by Proposition \ref{ferE}.
\end{proof}
Alternatively one can use \eqref{omKopForm} as the definition of the linear functional $\om^\vk$ in analogy to our definition of $\om^\wr$, and reorganize the above arguments accordingly to show that it is a transport plan, etc.

In analogy to Propositions \ref{verwBalKar} and \ref{verwBij} we have the next two propositions:
\begin{prop}
	\label{kopBalKar}
	Consider $\bz_2$-graded systems $\va$ and $\vb$,
	with copying maps $\vk_\ca$ and $\vk_\cb$ respectively, 
	as well as any transport plan $\om \in T^{\F}(\m , \n)$. 
	Then
	$$\va \om \vb \iff \va^\vk \om^\vk \vb^\vk.$$ 
\end{prop}
\begin{proof}
	By \eqref{Ekap} we clearly have
	$E_\om \circ \a = \b \circ E_\om$
	if and only if 
	$E_\om^\vk \circ \a^\vk = \b^\vk \circ E_\om^\vk$, 
	proving the proposition by \eqref{omkap} and Corollary \ref{ferBalEDirek}.
\end{proof}

\begin{prop}
	\label{kopBij}
	For $\bz_2$-graded systems $\va$ and $\vb$
	with copying maps 
	$\vk_\ca$ and $\vk_\cb$ respectively, 
	the map 
	$$
	\gt_\vk : T^{\F} ( \m , \n ) \to T^{\F} ( \m^\vk , \n^\vk ) 
	      : \om \mapsto \om^\vk
	$$
	and its restrictions
	$$
	\gt_\vk|_{T^{\F} ( \va , \vb )}  : 
	T^{\F} ( \va , \vb ) \to T^{\F} ( \va^\vk , \vb^\vk ),
	$$
	$$
	\gt_\vk|_{T^{\F}_\s ( \va , \vb )}  : 
	T^{\F}_\s ( \va , \vb ) \to T^{\F}_\s ( \va^\vk , \vb^\vk )
	$$
	and
	$$
	\gt_\vk|_{T^{\F}_{\s\s} ( \va , \vb )}  : 
	T^{\F}_{\s\s} ( \va , \vb ) \to T^{\F}_{\s\s} ( \va^\vk , \vb^\vk )
	$$
	are well-defined bijections.
\end{prop}
\begin{proof} 
	By Proposition \ref{OPkop} we know that $\gt_\vk$ is well-defined, as is the similarly defined map
	$T^{\F} ( \n^\vk , \m^\vk ) \to T^{\F} ( \m , \n ) : \psi \mapsto \psi^{\vk^{-1}}$,
	given by
	$
	\psi^{\vk^{-1}} 
	=
	\psi \circ ( \vk_\ca \ftp \vk_{\cb^\wr} ) ,
	$
	since $\vk_\ca^{-1}$ copies from $\ca^\wr$ to $\ca^{\wr\wr}$,
	giving an inverse for $\gt_\vk$.
	Thus $\gt_\vk$ is bijective.
	So is it's restriction
	$T^{\F} ( \va , \vb ) \to T^{\F} ( \va^\vk , \vb^\vk )$,
	using Proposition \ref{kopBalKar}.
	
	By \eqref{modGroepTransf} in the Appendix, we have
	$
	\s^{\m^\vk}_t = ( \s^\m_t )^\vk
	$,
	which, combined with Proposition \ref{kopBalKar}, tells us that
	$$
	  (\ca , \s^\m , \m) \om (\cb , \s^\n , \n) 
	   \iff
	   (\ca^\wr , \s^{\m^\vk} , \m^\vk) \om^\vk (\cb^\wr , \s^{\n^\vk} , \n^\vk),
	$$
	proving that $\gt_\vk|_{T^{\F}_\s ( \va , \vb )}$ is bijective. 
	
	Since $\m^\vk \circ \vk_\ca = \m$, we know (see the Appendix) that 
	$K_\ca J_\ca = J_\ca K_\ca$.
	Consequently,
	$$
	(\a^\s)^\vk = (\a^\vk)^\s ,
	$$
	through a direct computation in terms of the Hilbert space representations. Similarly for $\b$. Hence
	$$
	\va^\s \om \vb^\s
	\iff
	(\va^\vk)^\s \om^\vk (\vb^\vk)^\s,
	$$
	which implies that $\gt_\vk|_{T^{\F}_\s ( \va , \vb )}$ restricted to $T^{\F}_{\s\s} ( \va , \vb )$ indeed gives the last of the required bijections.
\end{proof}
Thus we have the following symmetries of our Wasserstein distances:
\begin{theorem} \label{kopSim} 
	For $\bz_2$-graded systems $\va$ and $\vb$ 
	with copying maps $\vk_\ca$ and $\vk_\cb$ respectively,
	we have
	$$
	W^{\F} ( \va^\vk , \vb^\vk ) =  W^{\F} (\va , \vb),
	$$
	$$
	W^{\F}_\s ( \va^\vk , \va^\vk ) =  W^{\F}_\s (\va , \vb)
	$$
	and
	$$
	W^{\F}_{\s\s} (\va^\vk , \vb^\vk) =  W^{\F}_{\s\s} (\va , \vb).
	$$
\end{theorem}
\begin{proof}
	This follows from the definitions of the Wasserstein distances and Proposition \ref{kopBij}, since 
	$
	I_{ \va^\vk , \vb^\vk} (\om^\vk)
	=
	I_{ \va , \vb} (\om) ,
	$
	as is easily verified from the involved definitions.
\end{proof}
We can now conclude this series of symmetries with the case of reversible systems (see Convention \ref{revConv} and Definition \ref{revDef}), 
which is of interest to us with regards to FDB. 
This requires the following simple point regarding cost in reversible systems.
\begin{prop}
	\label{omkKoördKoste}
	For reversible systems $\va$ and $\vb$, the coordinates
	$$ 
	k_\th = ( \th_\ca ( k_1^* ) , ... , \th_\ca ( k_d^* ) ) 
	\text{ and }
	l_\th = ( \th_\cb ( l_1^* ) , ... , \th_\cb ( l_d^* ) ) 
	$$
	give the same cost as $k$ and $l$. In particular, in terms of the notation
	$$
	\va_\th = ( \ca , \a , \m , k_\th ) \, ,
	$$
	we have
	$$
	W^{\F} ( \va_\th , \vb_\th ) =  W^{\F} (\va , \vb) \, ,
	$$
	and likewise for $W^{\F}_\s$ and $W^{\F}_{\s\s}$.
\end{prop}
\begin{proof}
	Keeping Definition \ref{I} of cost in mind, simply notice from Proposition \ref{omkOp} that 
	$
	\m ( \th_\ca (k_i^*)^* \th_\ca (k_i^*) )
	=
	\m ( k_i^* k_i )
	$,
	and from Corollary \ref{ferBalEDirek} together with Convention \ref{revConv} , that
	$
	\n ( E_\om ( \th_\ca (k_i^*) )^* \th_\cb ( l_i^* ) )
	=
	\n ( l_i^* E_\om (k_i) )
	$,
	which is all that is required.
\end{proof}

Thus our final pair of symmetries:
\begin{cor} \label{omkSim}
	For reversible $\bz_2$-graded systems $\va$ and $\vb$ 
	we have
	$$
	W^{\F}_\s ( \vb^\lp , \va^\lp ) =  W^{\F}_\s (\va , \vb)
	\quad \text{and} \quad 
	W^{\F}_{\s\s} ( \va^\lp , \vb^\lp ) =  W^{\F}_{\s\s} (\va , \vb).
	$$
\end{cor}
\begin{proof} 
	From the definitions and using $\m^\vk = \m^\wr$, we have
	$
	( \va^\lp )^\vk = \va_\th^\wr 
	$.
	Now apply Theorems \ref{verwSim} and \ref{kopSim}, as well as Proposition \ref{omkKoördKoste}.
\end{proof}

This symmetry, along with the triangle inequality, can be used in analogy to \cite[Subsection 8.1]{DSS} to put bounds on deviation of a reversible system $\va$ from FDB in terms of its distance from a system $\vb$ that satisfies FDB:
\begin{cor}
	For reversible systems $\va$ and $\vb$, with $\vb$ satisfying FDB, we have
\[
W^{\F}_{\sigma}(\mathbf{A},\mathbf{A}^{\leftarrow})
\leq
2W^{\F}_{\sigma}(\mathbf{A},\mathbf{B}) \, ,
\]
\[
W^{\F}_{\sigma}(\mathbf{A}^{\leftarrow},\mathbf{A})
\leq
2W^{\F}_{\sigma}(\mathbf{B},\mathbf{A})
\]
and
\[
W^{\F}_{\sigma\sigma}(\mathbf{A},\mathbf{A}^{\leftarrow})
\leq
2W^{\F}_{\sigma\sigma}(\mathbf{A},\mathbf{B})  .
\]
\end{cor}

Bounds of this nature serve as one motivation for studying symmetries of fermionic (and other quantum) Wasserstein distances. Also see \cite{GPTV, SV} for other work on symmetries, though in a very different context and for a different approach to quantum Wasserstein distance.

On the other hand, further work is still needed 
regarding the value of FDB as an equilibrium condition for fermionic systems, compared to the usual standard detailed balance condition. See \cite{D18} for initial steps.
Moreover, the bounds above warrant further investigation, for example in comparison to entropy production as a measure of deviation from quantum detailed balance \cite{FR10, FR, FR15b}.
In a different direction, it may also be worth to explore extending the theory of Wasserstein distances to more general gradings. See the development around the latter initiated in \cite{FV}.

\section*{Appendix on standard forms}
\label{Byl}


Here we briefly review and discuss some  points regarding Tomita-Takesaki (or modular) theory and standard forms of von Neumann algebras. A convenient formulation of the latter is as follows.

\begin{sf*}
\label{sform}
A $\s$-finite von Neumann algebra $\ca \subset B(\cg_\ca)$ represented in the algebra of bounded operators on a Hilbert space $\cg_\ca$ as indicated, is said to be in \textit{standard form} if we have the following:
A conjugate linear operator $J_{\ca}: \cg_{\ca} \to \cg_{\ca}$ on the Hilbert space $\cg_ {\ca}$, called a \emph{modular conjugation} for $\ca$, such that 
$$J_{\ca}^2=1,$$
$$J_{\ca}^{*}=J_{\ca}$$
and
$$J_{\ca}\ca J_{\ca}=\ca',$$
where $\ca'$ is the commutant of $\ca$ in $B(\cg_\ca)$.
In addition we have a closed convex cone $\mathcal{P}_{\ca} \subset \cg_{\ca}$ such that
$$J_{\ca}x=x$$
for all $x \in \cp_\ca$.
An element $x \in \mathcal{P}_{\ca}$ is cyclic for $\ca$ if and only if it is separating for $\ca,$ and every normal state $\mu$ on $\ca$ is given by
$$\mu(a)=\langle \L_\m , a\L_\m \rangle,$$
for a unique $\Lambda_{\mu} \in \mathcal{P}_{\ca}.$
\end{sf*}

	
For a treatment of the theory of standard forms, refer to \cite[Section 2.5.4]{BR1}, the original papers being \cite{Ar74, Con72, Con74, Ha75}.
Note that if $\mu$ is faithful, then $\Lambda_{\mu}$ is cyclic as it is separating.
Because of this, the definition above is ultimately equivalent to simply saying that the von Neumann algebra is in a cyclic representation given by any faithful normal state (see for example \cite[Definition III 2.6.1]{Bl}). However, the formulation above lists a number of key properties clarifying the nature of a standard form.

%


A standard form always exists. To obtain a standard form, one uses Tomita-Takesaki theory (originating in \cite{Tom}, properly developed in \cite{T70}), which we now outline very briefly (also see \cite[Section 2.5]{BR1}).

Given a faithful normal state $\mu$ on $\ca$ in its cyclic representation $(\cg_{\mu},\id_{\ca},\L_{\mu})$, that is
$\m(a) = \langle \L_\m , a\L_\m \rangle,$ 
we have a closed conjugate linear (in general unbounded) operator on $\ca\Lambda_ {\mu}$ 
determined by
$$S_\m a\L_\m := a^*\L_\m.$$
One considers the polar decomposition 
$$S_\m =J_\m\D_\m^{1/2},$$
with $\D_\m^{1/2}$ being self-adjoint and positive, while $J_\m$ is anti-unitary and self-adjoint.
We note that 
$
\D_\m^{it}\L_\m = \L_\m
$
for all real $t$, and we write
$$\s_t^\m(a) := \D_\m^{it} a \D_\m^{-it}$$
for all $a\in\ca$. 
The Tomita-Takesaki theorem tells us that
$$\s_t^\m(\ca) = \ca.$$
The one-parameter group given by 
$t \mapsto \s_t^\m$ 
is called the \textit{modular group} associated with the state $\m$. 

The requirements of a standard form above can now be shown to be met, including that $\cg_{\mu}$ and $J_{\mu}$ can be taken to be independent of the faithful normal state $\mu$, which is known as universality, and ultimately follows from the uniqueness of polar decompositions. 
This allows us to refer to $\cg_{\mu}$ as $\cg_{\ca}$ and to $J_{\mu}$ as $J_{\ca}$ in the standard from above.
We also use the notation
$$j_{\ca}(a):=J_{\ca}a^{*}J_{\ca}$$
for all $a$ in $B(\cg_\ca)$.

Since we also have the polar decomposition 
$$S_{\m'} =J_\ca \D_\m^{-1/2},$$
we see that $\ca'$ in effect has the same standard form as $\ca$, 
but 
with the modular groups of $\m$ and $\m'$ being related through
$$(\s^\m_t)' = \s^{\m'}_t.$$ 

For $\ca^\wr$ in the case of a $\bz_2$-grading $\g_\ca$ on the other hand, for any $a' \in \ca'$ we see that
$$
S_{\m^\wr} \g_\ca^{1/2}(a') \L_\m = \g_\ca^{1/2}(a')^*\L_\m \, ,
$$
from which it follows that
$$
   S_{\m^\wr} 
   = 
   g_\ca^{1/2} S_{\m'} g_\ca^{-1/2} 
   = 
   J_\ca g_\ca^{1/2} \D_{\m'}^{1/2} g_\ca^{-1/2}.
$$
That is, 
$J_{\ca^\wr} = J_\ca$ 
and
$\D_{\m^\wr}^{1/2}  = g_\ca^{1/2} \D_\m^{-1/2} g_\ca^{-1/2}$.
Through the latter we in turn obtain
\begin{equation}\label{verwMod}
(\s^\m_t)^\wr = \s^{\m^\wr}_t.
\end{equation}

If the $\bz_2$-graded von Neumann algebra $\cb$ is in standard form as well, and 
$\f : \ca \to \cb$
is a $*$-isomorphism such that 
$\n \circ \f = \m$
for faithful normal states $\m$ and $\n$ on $\ca$ and $\cb$ respectively, then $\f$ has a corresponding unitary \emph{GNS representation}
$U : \cg_\ca \to \cg_\cb$
extending
$$
Ua\L_\m = \f(a)\L_\n,
$$
and satisfying $\f(a) = UaU^*$.

Consequently, if $\f$ is even it is easily verified that
$$Ug_\ca = g_\cb U \, ,$$
thus $Ug_\ca^{1/2} = g_\cb^{1/2} U$ from the definition of $g_\ca^{1/2}$.

Irrespective of the presence of gradings or whether $\f$ even, we have
$$
S_\n \f(a) \L_\n
=
U S_\m a \L_\m
=
U S_\m U^* \f(a) \L_\n
$$
for all $a \in \ca$, hence
$
S_\n = U S_\m U^*,
$
thus 
$
J_\cb = U J_\ca U^*,
$
i.e.,
\begin{equation} \label{UJ=JU}
U J_\ca = J_\cb U \, ,
\end{equation}
due to the uniqueness of polar decomposition.
In particular, when $\cb = \ca$, it follows that 
$$
U J_\ca = J_\ca U \, ,
$$
even when $\n \neq \m$.
Consult \cite[Corollary 2.5.32]{BR1} for a presentation of a refinement in the standard form, used for the unitary representation of $\bz_2$-gradings in Section \ref{duality}.

If instead $\f : \ca \to \cb$ is a $*$-anti-isomorphism, then the same conclusions hold by essentially the same arguments in terms of the corresponding anti-unitary GNS representation $U$ obtained from
$$
Ua\L_\m = \f(a^*)\L_\n \, ,
$$
which gives $\f(a) = U a^* U^*$.

Returning to the $*$-isomorphism 
$\f : \ca \to \cb$,
the uniqueness of polar decomposition leads to 
$\D_\n^{1/2} = U \D_\m^{1/2} U^*$. Consequently
\begin{equation} \label{modGroepTransf}
	\s^\n_t = \f \circ \s^\m_t \circ \f^{-1} .
\end{equation}

%
%

\section*{Acknowledgements}
We thank Vito Crismale, Stefano Rossi and Elena Griseta for helpful discussions and being our hosts at Universit\`a degli Studi di Bari Aldo Moro during two visits in 2023, when much of this work was done.

\end{document}